\documentclass[11pt]{article}
  
\usepackage{cite}
\usepackage{graphicx}
\usepackage{amsmath}
\usepackage{amsthm}
\usepackage{yhmath} 

\usepackage{amsfonts} 
\usepackage{amssymb}
\usepackage{fullpage}
\usepackage{color}
\usepackage{latexsym}

\usepackage{ccicons }
\usepackage{cclicenses }
\usepackage{algorithm}
\usepackage[noend]{algorithmic}
\usepackage{multirow} 
\usepackage{amsmath}
\usepackage{yhmath}
\usepackage{amsthm}
\usepackage[utf8]{inputenc}
\usepackage{nameref}

\newcommand{\old}[1]{{}}

\newtheorem{theorem}{Theorem}
\newtheorem{corollary}{Corollary}
\newtheorem{lemma}{Lemma}

\newtheorem{observation}{Observation}

\newcommand{\D}{{{\cal{D}}}}

\renewcommand{\l}{{{\ell}}}

\newcommand{\tarc}{\mbox{\large$\frown$}}
\newcommand{\arc}[1]{\stackrel{\tarc}{#1}}

\title{Piercing Pairwise Intersecting Geodesic Disks by Five Points\footnote{This work was partially supported by Grant 2016116 from the United States -- Israel Binational Science Foundation.} }

\author{ A. Karim Abu-Affash\thanks{Software Engineering Department, Shamoon College of Engineering, Beer-Sheva 84100, Israel, {\tt abuaa1@sce.ac.il}.}
\and 
Paz Carmi\thanks{Department of Computer Science, Ben-Gurion University, Beer-Sheva 84105, Israel, {\tt carmip@cs.bgu.ac.il}. }
\and
Meytal Maman\thanks{Department of Computer Science, Ben-Gurion University, Beer-Sheva 84105, Israel, {\tt meytal.maman@gmail.com}.}
}

\begin{document}

\maketitle

\begin{abstract}
Given a simple polygon $P$ on $n$ vertices, and a set $\D$ of $m$ pairwise intersecting geodesic disks in $P$, we show that five points in $P$ are always sufficient to pierce all the disks in $\D$. This improves the previous bound of 14, obtained by Bose, Carmi, and Shermer~\cite{Bose21}.
\end{abstract}


\section{Introduction}

The problem of piercing geometric objects with as few points as possible has attracted the attention of researchers for the past century. The research so far has been focused on convex objects and disks in the plane. The most known result for piercing geometric objects with set of minimum cardinality, is known as Helly's theorem~\cite{Helly23, Helly30}, and works for convex sets in the plane. This theorem states the following: Given a set of $m$ convex objects in $\mathbb{R}^d$ such that $m>d+1$, if every $d+1$ of these objects have a point in common, then all of them have a point in common. 
This means that one point is sufficient to pierce all the objects.
This claim does not hold when the convex objects are only pairwise intersecting. However, for a set of disks in the plane, where every pair of disks intersects, it has proven by Danzer~\cite{Danzer86} and by Stacho~\cite{ Stacho65, Stacho814} that four points are sufficient to pierce all the disks. These proofs are not amenable to design efficient (subquadratic-time) algorithms for computing the piercing points. Recently, linear-time algorithms have been presented by Har-Peled et al.~\cite{HarPeled21} for computing five points that pierce $m$ pairwise intersecting disks, and by Carmi et al.~\cite{Carmi18} for computing four points.

Let $P$ be a simple polygon. A \emph{geodesic disk} $D$ with radius $r$ centered at a point $c \in P$ is the set of all points $x \in P$, such that the length of the shortest path from $x$ to $c$ is at most $r$. Bose et al.~\cite{Bose21} showed that for any set $\D$ of pairwise intersecting geodesic disks in $P$, 14 points are sufficient to pierce all the disks in $\D$ and these points can be computed in linear time. 
%
In this paper, we prove that five points are sufficient to pierce all the disks in $\D$, which improve the result of Bose et al.~\cite{Bose21}. More precisely, we prove the following theorem.
\begin{theorem}
Given a simple polygon $P$ on $n$ vertices, and a set $\D$ of $m$ pairwise intersecting geodesic disks in $P$, five points in $P$ are sufficient to pierce all the disks in $\D$.
\end{theorem}

\section{The Setup and Preliminaries}
For simplicity of presentation, we adapt some notation that appeared in~\cite{Bose21}.
Moreover, we use the convention that all indices are taken modulo the size of the set involved.
Let $P$ be a simple $n$-vertex polygon in the plane and let $v_1,v_2,...,v_{n}$ be its vertices sorted in clockwise order.
For two points $x,y\in P$, the geodesic (shortest) path from $x$ to $y$ is denoted as $\Pi(x,y)$ and its length is the sum of the lengths of its edges, and is denoted as $|\Pi(x,y)|$. A geodesic disk with radius $r\geq 0$ centered at a point $c\in P$ is the set $\{y\in P: |\Pi(c,y)|\leq r\}$.
A \emph{geodesic triangle} on three points $a,b,c \in P$, denoted by $\triangle(a,b,c)$, is a weekly-simple polygon whose boundary consists of the paths $\Pi(a,b)$, $\Pi(b,c)$, and $\Pi(a,c)$; see Figure~\ref{fig:geodesic}.
A \emph{pseudo triangle} is a simple polygon with three convex vertices.

A set $X=\{x_1,...,x_k\}$ of at least three points in $P$ is \emph{geodesically collinear} if there exist two points $x_i,x_j\in X$, such that $X \subset \Pi(x_i,x_j)$.
Given three points $a,b,c \in P$ that are not geodesically collinear, the paths $\Pi(a,b)$ and $\Pi(a,c)$ have a common subpath until they diverge at a point $a'$. Similarly, let $b'$ (resp., $c'$) be the point where the paths $\Pi(b,a)$ and $\Pi(b,c)$ (resp., the paths $\Pi(c,a)$ and $\Pi(c,b)$) diverge; see Figure~\ref{fig:geodesic}. Pollack et al.~\cite{Pollack89} observed that $\triangle(a',b',c')$ is a pseudo triangle. We refer to $\triangle(a',b',c')$ as the \emph{geodesic core} of $\triangle(a,b,c)$ and denote it by $\bigtriangledown(a,b,c)$.
Pollack et al.~\cite{Pollack89} observed the following observation.
\begin{observation}\label{obs:coreAngles}
Let $a,b$ and $c$ be three points in $P$. Then the geodesic core $\bigtriangledown(a,b,c)$ has only reflex angles along its boundary and the interior of this triangle is fully contained in $P$.
\end{observation}
\begin{figure}[htb]
    \centering
    \includegraphics[width=0.64\textwidth]{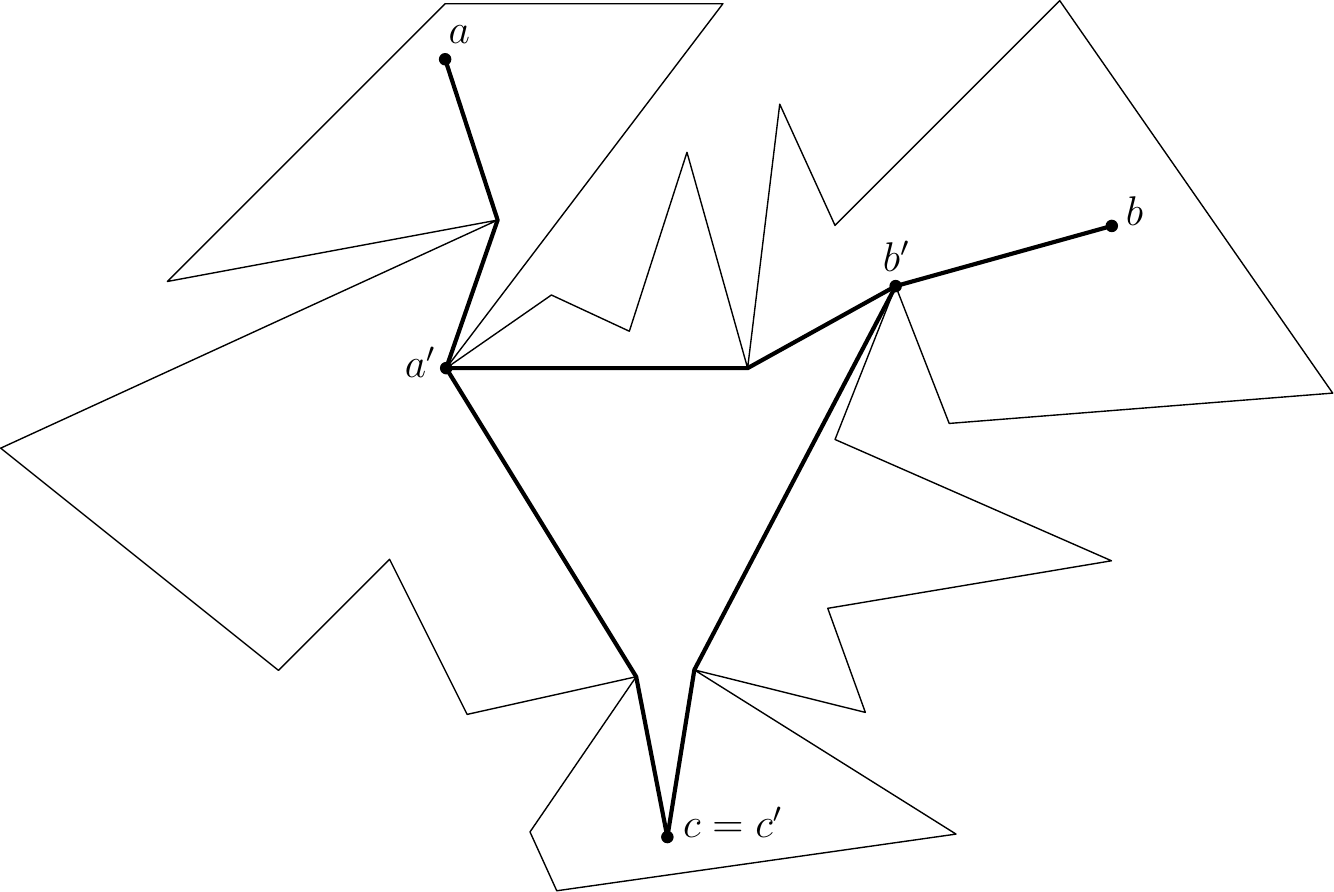}
    \caption{$\triangle(a,b,c)$ is a geodesic triangle. $\triangle(a',b',c)$ is a pseudo triangle.  $\bigtriangledown(a,b,c)=\triangle(a',b',c)$ is the geodesic core of $\triangle(a,b,c)$}.
    \label{fig:geodesic}
\end{figure}

Moreover, Pollack et al.~\cite{Pollack89} proved the following lemma about distances between a point and a geodesic path.
\begin{lemma}[\cite{Pollack89}] \label{lemma:observation}
Let $a,b$ and $c$ be three points in $P$. Let $g$ be the function defined on $\Pi(b,c)$, such that $g(x) = |\Pi(a,x)|$, for every point $x$ on $\Pi(b,c)$. Then, $g$ is a convex function with its maximum occurring either at $b$ or $c$. That is, $g(x) \le \max\{g(b),g(c)\}$, for every point $x$ on $\Pi(b,c)$.
\end{lemma}

The following observations follow from Lemma~\ref{lemma:observation}.
\begin{observation}\label{obs:observation}
Let $a$ and $b$ be two points, such that the segment $\overline{ab}$ is entirely contained in $P$. Then, any disk $D\in \D$ that contains both $a$ and $b$ must contain the segment $\overline{ab}$. 
\end{observation}

\begin{observation}\label{obs:observationTriangle}
Let $D$ be a geodesic disk in $\D$ with center $c \in P$, and let $a$ and $b$ be two points in $D$. Then, the pseudo-triangle $\triangle(c,a,b)$ is contained in $D$.
\end{observation}

\begin{observation}\label{obs:segmentEqual2}
Let $D$ be geodesic disk with center $c \in P$ and radius $r$. 
Let $q$ and $b$ be two points, such that $|\Pi(c,q)| + 1 \le r$, $|qb| \le 1$, and the segment $\overline{qb}$ is entirely contained in $P$. Then, $b$ is contained in $D$.
\end{observation}

Let $\D=\{D_1,D_2,...,D_m\}$ be a set of $m$ pairwise intersecting geodesic disks in $P$. For each $1\leq i \leq m$, let $c_i$ and $r_i$ denote the center and the radius of $D_i$, respectively.
The set $\D$ is called \emph{Helly} if there is a point that pierces all the disks in $\D$. 
For a point $x\in P$, we define a function $f(x)=y$ to be the smallest radius of a geodesic disk centered at $x$ that intersects all the disks in $\D$. 
A disk $D$ with radius $r$ centered at $c$ is called \emph{minimal} with respect to $\D$ if every point $x$ in the close neighborhood of $c$ in $P$ has $f(x)>r$.
Let $D^*$ be the disk with center $c^*$ that minimizes $f(c^*)$, and let $r^*=f(c^*)$ be its radius.
Bose et al.~\cite{Bose21} proved the following lemma regarding the properties of $D^*$.
\begin{lemma}[\cite{Bose21}] \label{lemma:thangentDisks}
If $\D$ is not Helly, then $D^*$ satisfies the following properties:
\begin{itemize}
	\item $r^* > 0$;
	\item $D^*$ does not intersect the boundary of $P$;
  \item $D^*$ is tangent to at least 3 geodesic disks $D_1, D_2, D_3$ in $\D$ at 3 distinct points $t_1, t_2, t_3$, respectively;
  \item $c^*$ is contained in the interior of $\triangle(t_1, t_2, t_3)$; and
  \item $D^*$ does not intersect the boundary of the geodesic core $\bigtriangledown(c_1,c_2,c_3)$, where $c_1,c_2,c_3$ are the centers of $D_1, D_2, D_3$, respectively.
\end{itemize}
\end{lemma} 

Assume, w.l.o.g., that $r^* =1$ and that $c^*$ is located at the origin $(0,0)$.
Let $D_1, D_2, D_3$ be the three geodesic disks from Lemma~\ref{lemma:thangentDisks} that are tangent to $D^*$ at the points $t_1, t_2, t_3$, respectively. 
For each $i\in \{1,2,3\}$, let $\l_i$ be the line that is tangent to $D_i$ and passes through $t_i$; see Figure~\ref{fig:emptyRegions}.
Let $m_{i,j}$ be the intersection point between the lines $\ell_i$ and $\ell_j$, for every distinct $i,j \in \{1,2,3\}$.
Assume, w.l.o.g., that $\ell_1$ is horizontal and the angle $\angle(m_{1,2},m_{2,3},m_{3,1})$ is the largest in the triangle $\triangle(m_{1,2},m_{2,3},m_{3,1})$; see Figure~\ref{fig:emptyRegions}.
\begin{figure}[htb]
    \centering
    \includegraphics[width=0.66\textwidth]{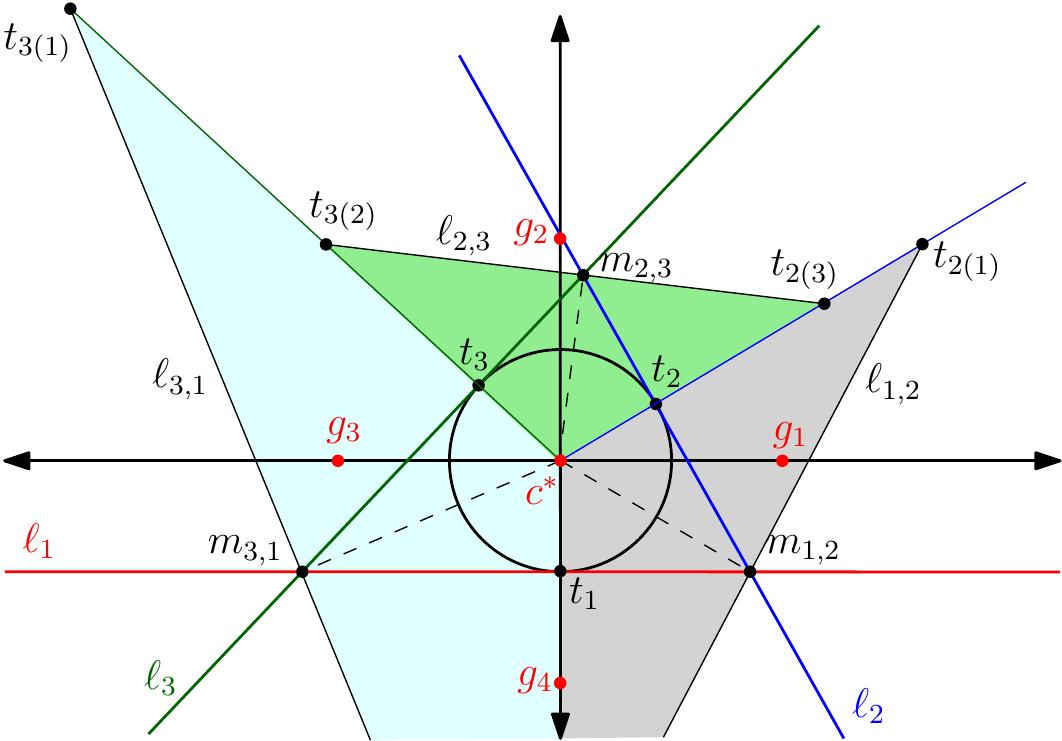}
    \caption{The smallest disk $D^*$ is located at the origin. $\ell_1,\ell_2,$ and $\ell_3$ are the tangent lines between $D^*$ and the disks $D_1, D_2$, and $D_3$, respectively. The path $\Pi(c_i,c_j)$ does not intersect $\triangle(t_{i(j)},c^*,t_{j(i)})$, for any distinct $i,j \in \{1,2,3\}$.}
    \label{fig:emptyRegions}
\end{figure}

For two points $p$ and $q$, let $\overline{pq}$ denote the line segment connecting them. For every distinct $i,j \in \{1,2,3\}$, let $\ell_{i,j}$ be the line passing through $m_{i,j}$ perpendicular to $\overline{c^*m_{i,j}}$; see Figure~\ref{fig:emptyRegions}.
Let $t_{i(j)}$ be the intersection point between $\ell_{i,j}$ and the line passing through $\overline{c^*t_i}$. The following lemma was proven in \cite{Bose21}.
\begin{lemma}[\cite{Bose21}] \label{lemma:emptyRegions}
The path $\Pi(c_i,c_j)$ does not intersect $\triangle(t_{i(j)},c^*,t_{j(i)})$, for any distinct $i,j \in \{1,2,3\}$.
\end{lemma}

Let $g_1$, $g_2$, $g_3$ and $g_4$ be the points located at the coordinates $(2,0),(0,2),(-2,0)$, and $(0,-2)$, respectively. The following corollary follows from Lemma~\ref{lemma:emptyRegions} and the assumption that $\ell_1$ is horizontal and the angle $\angle(m_{1,2},m_{2,3},m_{3,1})$ is the largest in the triangle $\triangle(m_{1,2},m_{2,3},m_{3,1})$.
\begin{corollary} \label{cor:emptyTriangles}
The polygon $P$ does not intersect the triangles $\triangle(g_1,c^*,g_4$) and $\triangle(g_3,c^*,g_4)$.
\end{corollary}

For a point $p \in P$, let $x(p)$ and $y(p)$ denote the $x$-coordinate and the $y$-coordinate of $p$, respectively.
We divide the plane into 4 quadrants $Q_1$, $Q_2$, $Q_3$, and $Q_4$ as follows; see Figure~\ref{fig:quarters}.
\begin{itemize}
	\item $Q_1 = \{p \in \mathbb{R}^2 : x(p) \ge 0 \ \text{ and } \ y(p) \ge 0 \}$;
	\item $Q_2 = \{p \in \mathbb{R}^2 : x(p) \le 0 \ \text{ and } \ y(p) \ge 0 \}$;
	\item $Q_3 = \{p \in \mathbb{R}^2 : x(p) \le 0 \ \text{ and } \ y(p) \le 0 \}$; and
	\item $Q_4 = \{p \in \mathbb{R}^2 : x(p) \ge 0 \ \text{ and } \ y(p) \le 0 \}$.
\end{itemize}
\begin{figure}[htb!]
    \centering
    \includegraphics[width=0.59\textwidth]{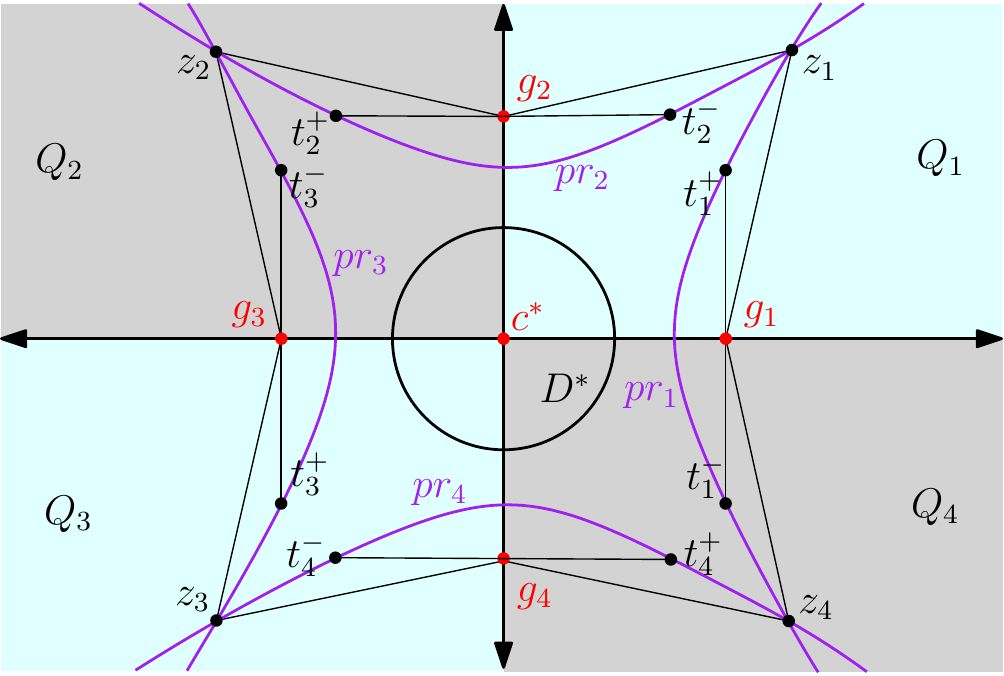}
    \caption{The quadrants $Q_i$ and the points $z_i$, for each $1 \le i \le 4$. 
    The segments $\overline{z_ig_i}$ and $\overline{z_{i-1}g_i}$ are entirely contained in the region bounded by the parabola $pr_i$ (depicted in purple).}
    \label{fig:quarters}
\end{figure}

For each $i \in \{1,2,3,4\}$, let $z_i \in Q_i$ be the point whose distance from $g_i, g_{i+1}$ and the boundary of $D^*$ is equal; see Figure~\ref{fig:quarters}. 
The computation of the points $z_i$ is not involved. For example, we compute $z_1$ by solving the following equations system:
    \begin{align*}
            &|z_1g_1| = \sqrt{(x(z_1)-2)^2+y(z_1)^2} = d \, , \\
            &|z_1g_2| = \sqrt{x(z_1)^2+(y(z_1)-2)^2} = d \, ,  \\
            &|z_1c^*| = \sqrt{x(z_1)^2+y(z_1)^2} = d+1 \, . 
    \end{align*}
This implies that $z_1= (a,a)$, $z_2= (-a,a)$, $z_3= (-a,-a)$, and $z_4= (a,-a)$, where $a=\frac{3}{4-2\sqrt{2}}\approx 2.56$.
For $1 \le i \le 4$, let $pr_i$ be the region bounded by the parabola that contains all the points that are closer to $g_i$ than to $D^*$; see Figure~\ref{fig:quarters}. 
Moreover, let $t^+_1=(2,1.5)$, $t^-_1=(2,-1.5)$, $t^+_2=(-1.5,2)$, $t^-_2=(1.5,2)$, $t^+_3=(-2,-1.5)$, $t^-_3=(-2,1.5)$, $t^+_4=(1.5,-2)$, and $t^-_4=(-1.5,-2)$. 
The following observations follow from the definition of $pr_i$.
\begin{observation}\label{obs:prabola}
Let $p$ be a point on $\overline{z_ig_i}$ or on $\overline{z_{i-1}g_i}$, where $i \in \{1,2,3,4\}$. Then, for any point $q$ on the boundary of $D^*$, we have $|pg_i| \le |pq|$.  
\end{observation}

\begin{observation}\label{obs:prabola1}
Let $p$ be a point on $\overline{t^+_it^-_i}$, where $i \in \{1,2,3,4\}$. Then, for any point $q$ on the boundary of $D^*$, we have $|pg_i| \le |pq|$.  
\end{observation}

Let $D\in \D$ be a disk with center $c$ and radius $r$. 
Throughout the rest of the paper, we use the following notations.
For each $i \in \{1,2,3\}$, let $q_i$ be the intersection point of the path $\Pi(c,c_i)$ with the boundary of $D_i$. 
Let $q^*$ be the intersection point of the path $\Pi(c,c^*)$ with the boundary of $D^*$.
Thus, $|\Pi(c , q^*)| \leq r$ and $|\Pi(c , q_i)| \leq r$, for each $i \in \{1,2,3\}$.  
Let $c'$ be the point on $\Pi(c,c^*)$, such that the edge $(c',c^*)$ is the last edge in $\Pi(c,c^*)$. That is, $c'$ is the first point on $\Pi(c,c^*)$ that is visible from $c^*$.
Finally, let $\alpha_2$ (resp., $\alpha_3$) be the acute angle between $\l_2$ (resp., $\l_3$) and the $x$-axis; see Figure~\ref{fig:alpha temp}.
%

\begin{observation}\label{lemma:ell3Angle}
If the polygon $P$ intersects the segment $\overline{z_4g_1}$ or $\overline{z_4g_4}$, then $ \alpha_2 > \frac{\pi}{5}$; see Figure~\ref{fig:alpha temp}. Similarly, if the polygon intersects the segment $\overline{z_3g_3}$ or $\overline{z_3g_4}$, then $\alpha_3 > \frac{\pi}{5}$.
\end{observation}
\begin{proof}
By Lemma~\ref{lemma:emptyRegions}, the polygon does not intersect the triangle $\triangle(t_{2(1)},c^*,t_{1(2)})$; see Figure~\ref{fig:alpha temp}.
Using a simple geometric calculation, for $\alpha_2 = \frac{\pi}{5}$, the acute angle between $\ell_{1,2}$ and the $x$-axis is $\beta = \frac{\pi - \alpha_2}{2} = \frac{2\pi}{5}$ and the coordinates of $m_{1,2}$ are $(\frac{\cos{\alpha_2} + 1}{\sin{\alpha_2}} , -1)$. Thus, for $\alpha_2 = \frac{\pi}{5}$, $\l_{1,2}$ passes through $z_4$. 
Therefore, for $0 < \alpha_2 \le \frac{\pi}{5}$, the point $z_4$ is contained in the triangle $\triangle(t_{2(1)},c^*,t_{1(2)})$, and, the polygon cannot intersect the segment $\overline{z_4g_1}$.
\end{proof}

    \begin{figure}[htb]
    \centering
    \includegraphics[width=0.47\textwidth]{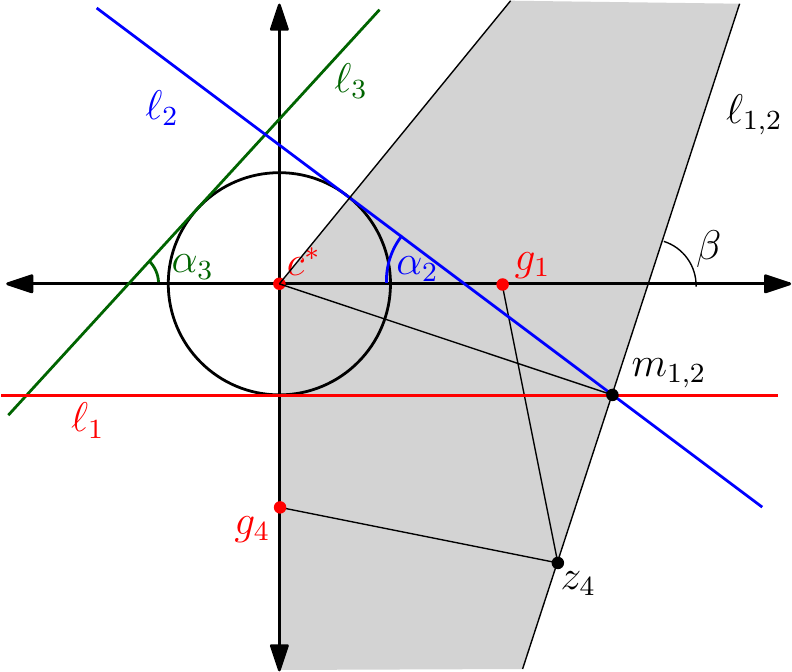}
    \caption{For $0 < \alpha_2 \le \frac{\pi}{5}$, the polygon cannot intersect the segment $\overline{z_4g_1}$.}
    \label{fig:alpha temp}
    \end{figure}

In the following lemma, we show that, for each $i \in \{1,2,3,4\}$, if the polygon does not intersect the segments $\overline{z_ig_i}$ and $\overline{z_ig_{i+1}}$, then every disk $D \in \D$ with $c'$ in $Q_i$ is pierced by at least one of the points $c^*$, $g_i$ or $g_{i+1}$. 

\begin{lemma} \label{lemma:NoIntersection}
Let $D \in \D$ be a disk with $c'$ in $Q_i$, where $i\in\{1,2,3,4\}$.
If the polygon does not intersect the segments $\overline{z_ig_i}$ nor $\overline{z_ig_{i+1}}$, then $D$ contains at least one of the points $c^*$, $g_i$ or $g_{i+1}$.
\end{lemma}
\begin{proof}
Let $c$ and $r$ be the center and the radius of $D$, respectively.
We distinguish between two cases: 
\noindent
\\{\bf Case~1:} The path $\Pi(c,c^*)$ intersects $\overline{z_ig_i}$ or $\overline{z_ig_{i+1}}$ at a point $p$.
    Assume w.l.o.g., $\Pi(c,c^*)$ intersects $\overline{z_ig_i}$; see Figure~\ref{fig:Q1 NoIntersection} (for $i=1$).
    By Observation~\ref{obs:prabola}, we have $|pg_i| \le |pq^*|$. Moreover, since the polygon does not intersect $\overline{z_ig_i}$, 
    we have $|\Pi(c,g_i)| \le |\Pi(c,p)| + |pg_i| \leq |\Pi(c,p)| + |pq^*| = |\Pi(c,q^*)| \le r$.
    Therefore, $D$ contains $g_i$.
 \begin{figure}[htb!]
    \centering
    \includegraphics[width=0.39\textwidth]{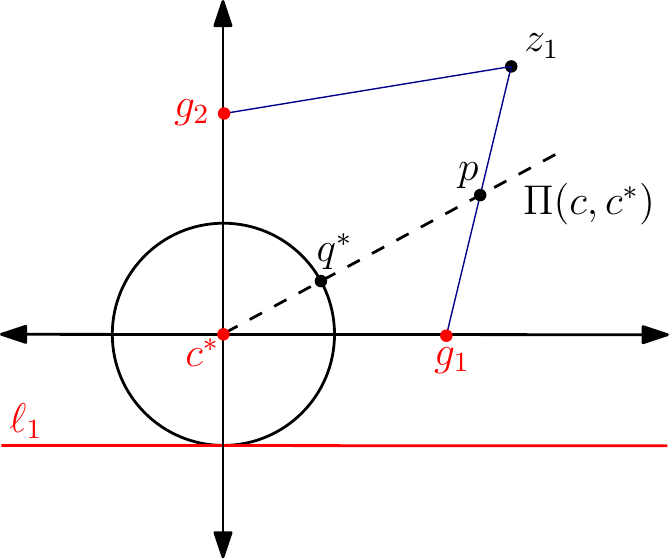}
    \caption{$c' \in Q_1$ and the path $\Pi(c',c^*)$ intersects $\overline{z_1g_1}$. }
    \label{fig:Q1 NoIntersection}
 \end{figure}

\noindent
{\bf Case~2:} The path $\Pi(c,c^*)$ does not intersect $\overline{z_ig_i}$ nor $\overline{z_ig_{i+1}}$.
We prove this case for $i=1$; the proof of the other cases are symmetric.
Consider the path $\Pi(c,c_1)$ and notice that it intersects the $x$-axis at a point $q$.
Since $|\Pi(q,q_1)|\geq 1$, we have $|\Pi(c,q)| + 1 \le |\Pi(c,q)| + |\Pi(q,q_1)| = |\Pi(c,q_1)| \leq r$.
By the case assumption, and by the fact that the polygon does not intersect $\overline{z_1g_1}$ nor $\overline{z_1g_2}$, $\Pi(c,c_1)$ has a vertex $p$ inside the quadrilateral defined by $c^*,g_1,z_1,g_2$, such that the polygon does not intersect the segment $\overline{pq}$; see Figure~\ref{fig:Q1 NoIntersection 2}.
Hence, $x(c^*) \leq x(q)\leq x(z_1)$ and $|\Pi(c,q)| = |\Pi(c,p)|+ |pq|$.  
Moreover, $|c^*g_1|=2$, and, by Corollary~\ref{cor:emptyTriangles}, the polygon does not intersect $\overline{c^*g_1}$.
\begin{itemize}
	\item If $x(c^*) \leq x(q)\leq x(g_1)$, then, since $|c^*g_1|=2$, we have $|c^*q|\le 1$ or $|qg_1| \le 1$, and by Observation~\ref{obs:segmentEqual2}, $D$ contains at least one of the points $c^*$ or $g_1$; see Figure~\ref{fig:Q1 NoIntersection 2}(a).
	\item If $x(g_1) < x(q)\leq x(z_1)$, then, since $x(g_1)=2$ and $x(z_1) < 3$, we have $|qg_1| \le 1$, and thus $|\Pi(c,g_1)| \le |\Pi(c,p)|+ |pq| + |qg_1| < |\Pi(c,p)|+ |pq| + 1 = |\Pi(c,q)| + 1 \le r$; see Figure~\ref{fig:Q1 NoIntersection 2}(b). Therefore, $D$ contains $g_1$. 
\end{itemize}
 
Notice that, by Corollary~\ref{cor:emptyTriangles}, the polygon does not intersect $\overline{c^*g_1}$, $\overline{c^*g_3}$, nor $\overline{c^*g_4}$. Thus, for $i=3$ and $i=4$, in Case~2, we have $p=c=c'$ is $D$'s center.
\end{proof}
\begin{figure}[htb]
    \centering
    \includegraphics[width=0.84\textwidth]{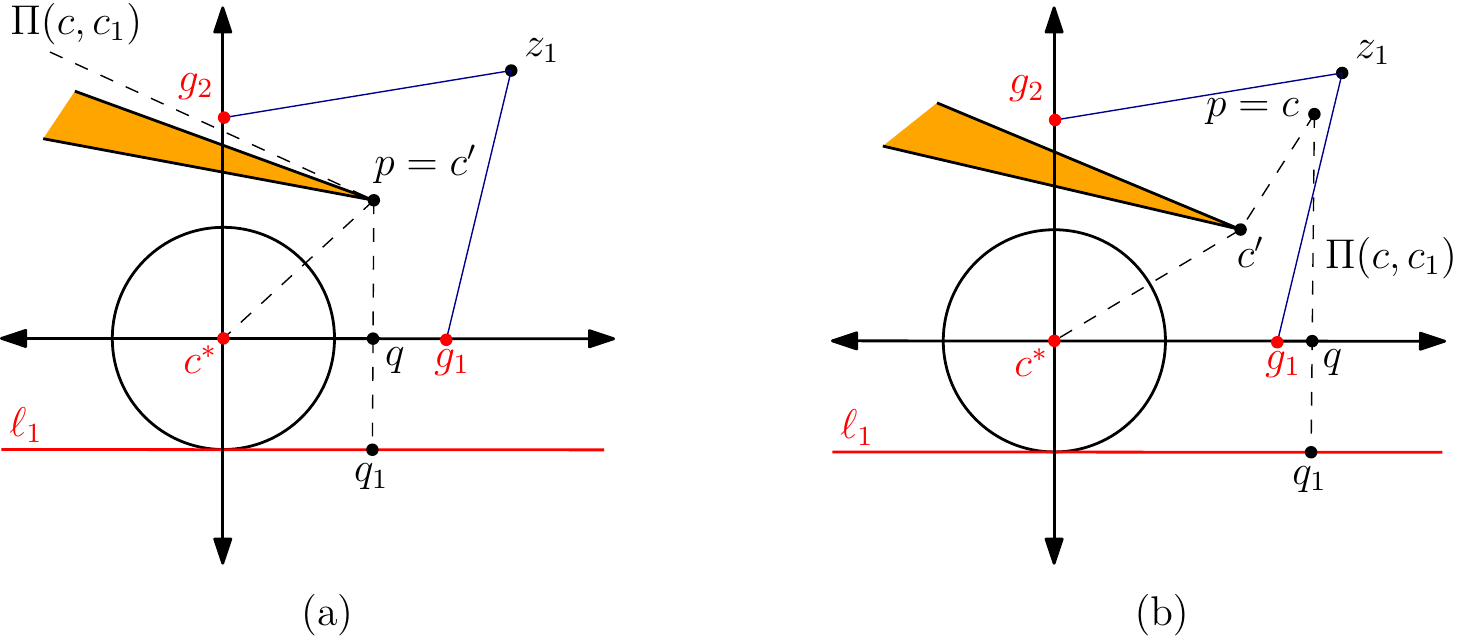}
    \caption{The path $\Pi(c',c^*)$ does not intersect $\overline{z_1g_1}$ nor $\overline{z_1g_2}$, (a) $p=c$. (b) $p \neq c$.}
    \label{fig:Q1 NoIntersection 2}
\end{figure} 

In the following, we define eight points $g^+_i \in Q_i$ and $g^-_i \in Q_{i-1}$, for each $i \in \{1,2,3,4\}$, and we prove some lemmas regarding these points. 
For each $i \in \{1,2,3,4\}$, let $m^+_i$ (resp., $m^-_i$) be the tangent line to $D^*$ that passes through $g_i$ and has a positive (resp., negative) slope; see Figure~\ref{fig:proc} (for an illustration of $m^+_1$ and $m^-_1$). 

The point $g^+_1$ is defined as follows.
Let $D'$ be the disk of radius 1 centered at the point $(1,0)$.
We sweep with a line $\l$ that is tangent to $D^*$ in counterclockwise order starting with $\l = m^+_1$ and we stop when $\l$ intersects either $D_3$ or the polygon inside the quadrilateral defined by $c^*,g_1,z_1,g_2$; see Figure~\ref{fig:proc}(a). 
Let $u$ be the intersection point of $\l$ with $D'$ in $Q_1$ when we stop the sweeping.
We also sweep upwards with a horizontal line $\l_h$ that passes through the point $c^*$, and stop when $\l_h$ intersects the polygon inside $D'$, or when $\l_h$'s $y$-coordinate is 1; see Figure~\ref{fig:proc}(b). 
Let $w$ be the intersection point of $\l_h$ with $D'$ in $Q_1$ when we stop the sweeping.
We set $g^+_1$ as the lowest point among $u$ and $w$.
\begin{figure}[bht]
\centering
\includegraphics[width=0.88\textwidth]{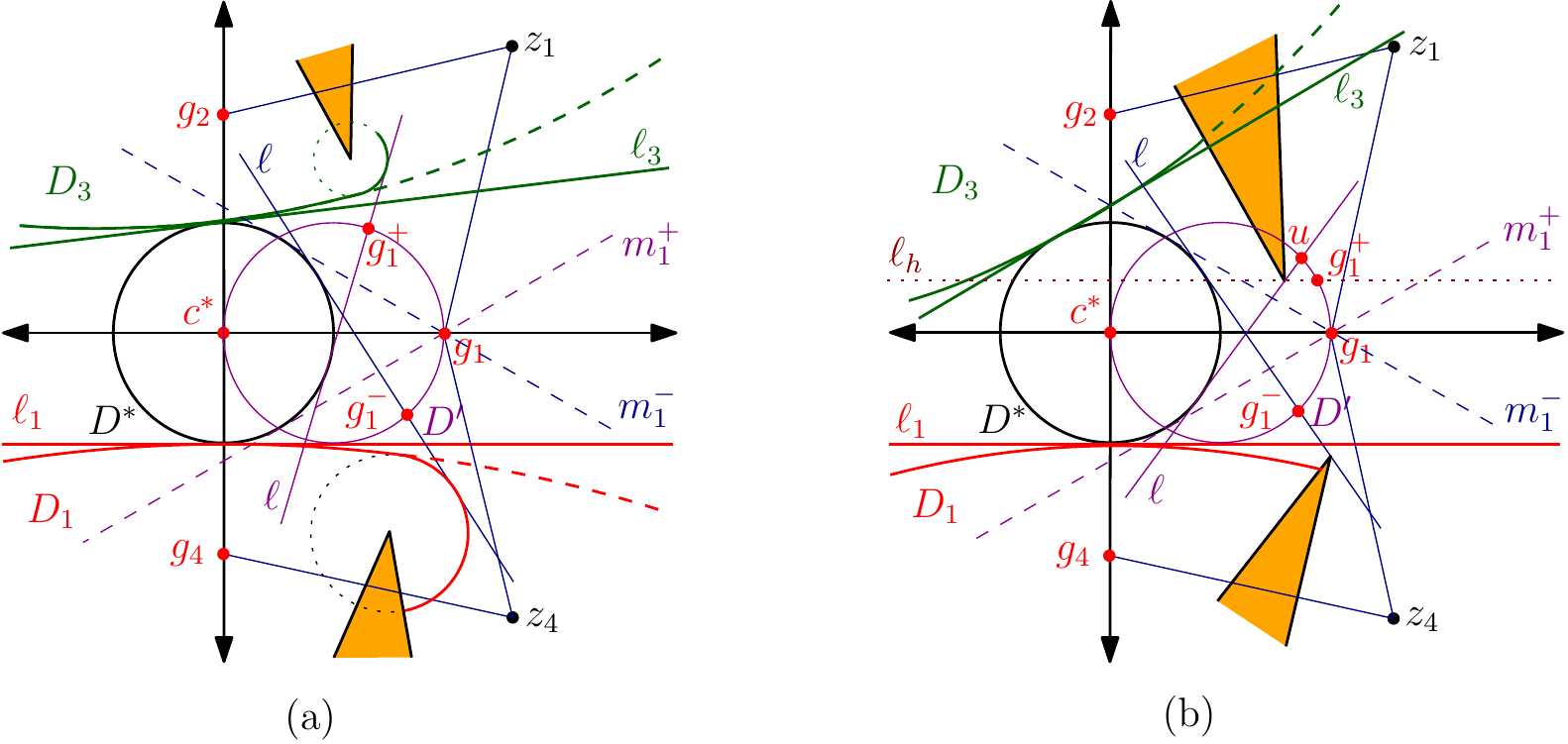}
\caption{Defining $g^+_1$ and $g^-_1$. (a) $g^+_1$ is defined by the intersection of $m^+_1$ with $D_3$ and $g^-_1$ is defined by the intersection of $m^-_1$ with $D_1$. (b) $g^+_1$ is defined by the intersection of $\l_h$ with the polygon inside $D'$ and $g^-_1$ is defined by the intersection of $m^-_1$ with the polygon outside $D'$.}
\label{fig:proc}
\end{figure}

The point $g^-_1$ is defined as follows. 
We sweep with a line $\l$ that is tangent to $D^*$ in clockwise order starting with $\l = m^-_1$ and we stop when $\l$ intersects either $D_1$ or the polygon inside the quadrilateral defined by $c^*,g_1,z_4,g_4$; see Figure~\ref{fig:proc}.
Let $u$ be the intersection point of $m^-_1$ with $D'$ in $Q_4$ when we stop the sweeping.
We also sweep downwards with a horizontal line $\l_h$ that passes through the point $c^*$, and stop either when $\l_h$ intersects the polygon inside $D'$, or when $\l_h$'s $y$-coordinate is -1. 
Let $w$ be the intersection point of $\l_h$ with $D'$ in $Q_4$ when we stop the sweeping.
We set $g^-_1$ as the highest point among $u$ and $w$.

We define $g^+_2 = (-1,1)$ and $g^-_2 = (1,1)$, and we define $g^+_3$, $g^-_3$, $g^+_4$, and $g^-_4$ similarly to $g^+_1$ and $g^-_1$, where the sweeping line $\ell$ starts with  $m^+_3$, $m^-_3$, $m^+_4$, and $m^-_4$, respectively. For $g^+_3$ and $g^-_3$, $D'$ is centered at $(-1,0)$ and, for $g^+_4$ and $g^-_4$, $D'$ is centered at $(0,-1)$.

\begin{lemma}\label{lemma:c*g'1}
Let $D\in \D$ be a disk centered at $c$ with radius $r$ and let $g'_i \in \{g_i,g^+_i,g^-_i\}$, for each $i \in \{1,2,3,4\}$. 
\begin{enumerate}
    \item[(i)] If $c' \in Q_1$, and $\Pi(c,c_1)$ intersects the $x$-axis at a point $q$ with $x(c^*) \le x(q) \le x(g'_1)$, then $|\Pi(c,c^*)|\le r$ or $|\Pi(c,g'_1)|\le r$; see Figure~\ref{fig:c*g1g2g3g4}(a).
    \item[(ii)] If $c' \in Q_2$ and $\Pi(c,c_1)$ intersects the $x$-axis at a point $q$ with $x(g'_3) \le x(q) \le x(c^*)$, then $|\Pi(c,c^*)|\le r$ or $|\Pi(c,g'_3)|\le r$; see Figure~\ref{fig:c*g1g2g3g4}(b).
    \item[(iii)] If  $c' \in Q_3$  and $\Pi(c,c_2)$ intersects the $x$-axis at a point $q$ with $x(g'_3) \le x(q) \le x(c^*)$, then $|\Pi(c,c^*)|\le r$ or $|\Pi(c,g'_3)|\le r$; see Figure~\ref{fig:c*g1g2g3g4}(c). 
    \item[(iv)] If $c' \in Q_3$  and $\Pi(c,c_2)$ intersects the $y$-axis at a point $q$ with $y(g'_4) \le y(q) \le y(c^*)$, then $|\Pi(c,c^*)|\le r$ or $|\Pi(c,g'_4)|\le r$; see Figure~\ref{fig:c*g1g2g3g4}(d).      
    \item[(v)] If $c' \in Q_4$ and $\Pi(c,c_3)$ intersects the $x$-axis at a point $q$ with $x(c^*) \le x(q) \le x(g'_1)$, then $|\Pi(c,c^*)|\le r$ or $|\Pi(c,g'_1)|\le r$; see Figure~\ref{fig:c*g1g2g3g4}(e).
   \item[(vi)] If $c' \in Q_4$ and $\Pi(c,c_3)$ intersects the $y$-axis at a point $q$ with $y(g'_4) \le y(q) \le y(c^*)$, then $|\Pi(c,c^*)|\le r$ or $|\Pi(c,g'_4)|\le r$; see Figure~\ref{fig:c*g1g2g3g4}(f).
\end{enumerate}
\end{lemma}
\begin{figure}[H]
\centering
\includegraphics[width=0.76\textwidth]{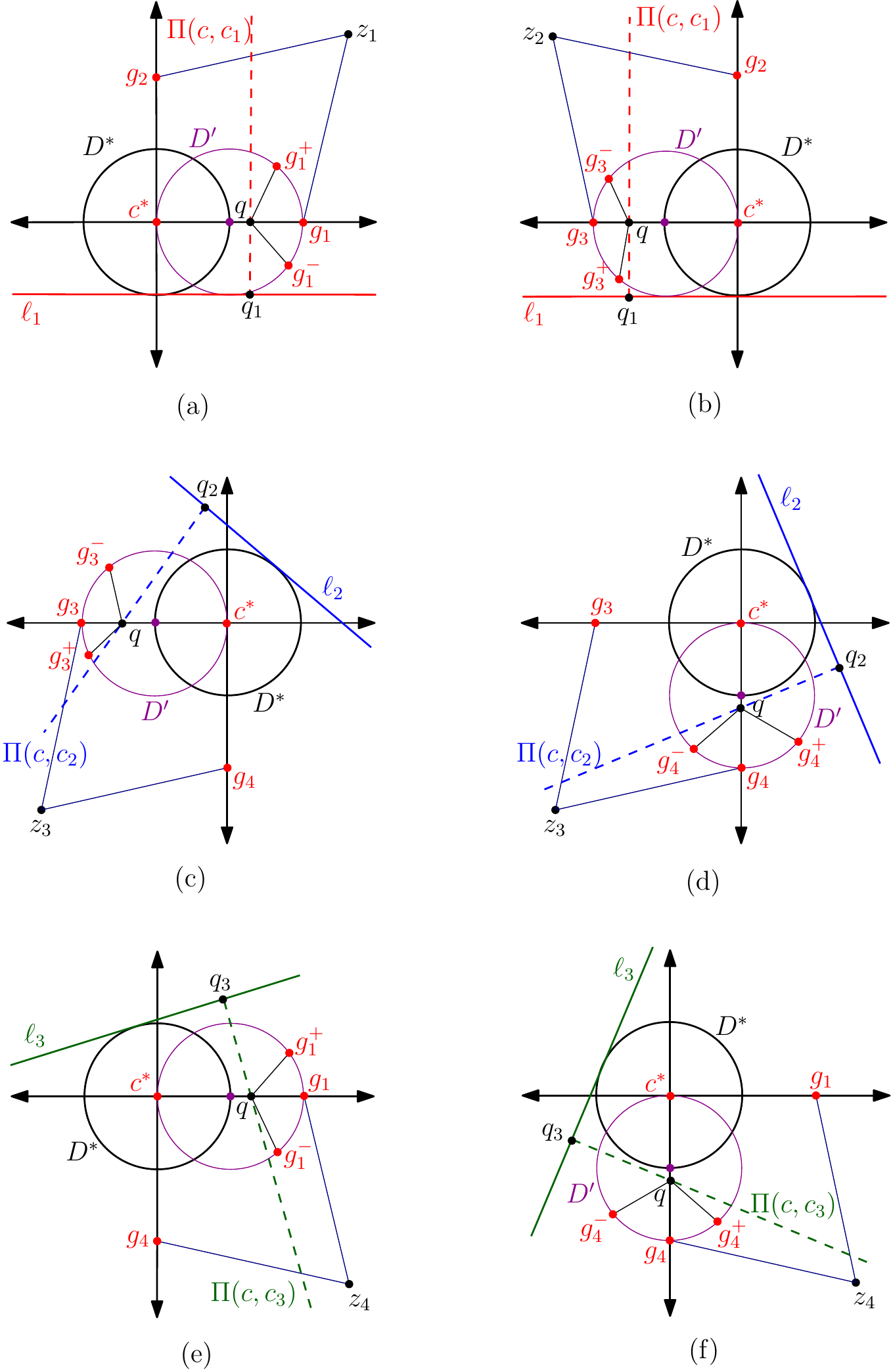}
\caption{Illustration of Lemma~\ref{lemma:c*g'1}: (a) Item (i), (b) Item (ii), (c) Item (iii), (d) Item (iv), (e) Item (v), and (f) Item (vi).}
\label{fig:c*g1g2g3g4}
\end{figure}

\begin{proof}
We prove Item (i), the proofs of the other five items are symmetric.

Notice that $|\Pi(c,q)| + 1 \le r$. Moreover, by Corollary~\ref{cor:emptyTriangles}, the polygon does not intersect $\overline{c^*g_1}$; see Figure~\ref{fig:c*g1}.
\begin{itemize}
    \item If $x(c^*)\leq x(q) \leq 1$, then, $|c^*q|\le 1$; see Figure~\ref{fig:c*g1}(a). 
    Thus, $|\Pi(c,c^*)|\le |\Pi(c,q)| + |qc^*| \le |\Pi(c,q)| + 1 \le r$.
    \item If $1 < x(q) \leq x(g'_1)$, then $|qg'_1| \le 1$; see Figure~\ref{fig:c*g1}(b). Moreover, by the definition of $g'_1$, the polygon does not intersect $\overline{qg'_1}$. Thus, $|\Pi(c,g'_1)|\le |\Pi(c,q)| + |qg'_1| \le |\Pi(c,q)| + 1 \le r$. 
\end{itemize}
\vspace{-0.7cm}
\end{proof}
\begin{figure}[htb]
\centering
\includegraphics[width=0.85\textwidth]{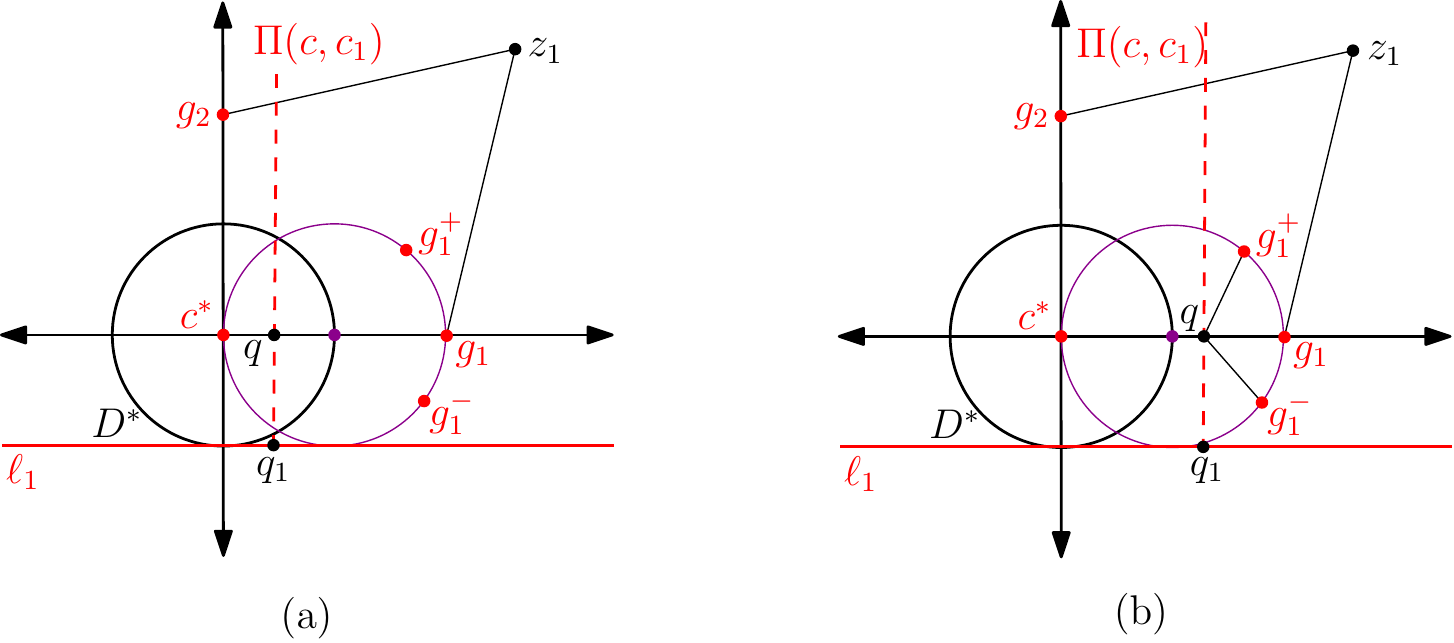}
\caption{Illustration of the proof of Lemma~\ref{lemma:c*g'1}, Item (i): (a) $x(c^*)\leq x(q) \leq 1$, and (b) $1 \leq x(q) \leq x(g'_1)$.}
\label{fig:c*g1}
\end{figure}

\begin{lemma}\label{lemma:g1+ right z1g2}
Let $D\in \D$ be a disk centered at $c$ with radius $r$. 
\begin{enumerate}
    \item[(i)]  If $c' \in Q_1$, $\Pi(c,c_1)$ intersects the $x$-axis at a point $q$ with $x(q) > x(g^+_1)$, $\Pi(c,c^*)$ intersects the segment $\overline{z_1g_2}$ and the polygon intersects the segment $\overline{z_2g_2}$, then $|\Pi(c,g^+_1)|\leq r$; see Figure~\ref{fig:Lemma6}(a).
    \item[(ii)] If $c' \in Q_1$, $\Pi(c,c_1)$ intersects the $x$-axis at a point $q$ with $x(q) > x(g^+_1)$, $\Pi(c,c^*)$ intersects the segment $\overline{z_1g_2}$ and the polygon does not intersect the segment $\overline{z_2g_2}$, then $|\Pi(c,g^+_1)|\leq r$ or $|\Pi(c,g_2)|\leq r$; see Figure~\ref{fig:Lemma6}(b).
    \item[(iii)] If $c' \in Q_3$, $\Pi(c,c_2)$ intersects the $y$-axis at a point $q$ with $y(q) < y(g^-_4)$, and $\Pi(c,c^*)$ intersects the segment $\overline{z_3g_3}$, then, $|\Pi(c,g^-_4)|\leq r$ or $|\Pi(c,g_3)|\leq r$; see Figure~\ref{fig:Lemma6}(c).
    \item[(iv)] If $c' \in Q_4$,  $\Pi(c,c_3)$ intersects the $x$-axis at a point $q$ with $x(q) > x(g^-_1)$, and $\Pi(c,c^*)$ intersects the segment $\overline{z_4g_4}$, then, $|\Pi(c,g^-_1)|\leq r$ or $|\Pi(c,g_4)|\leq r$; see Figure~\ref{fig:Lemma6}(d).
\end{enumerate}
\end{lemma}
\begin{figure}[htb]
\centering
\includegraphics[width=0.9\textwidth]{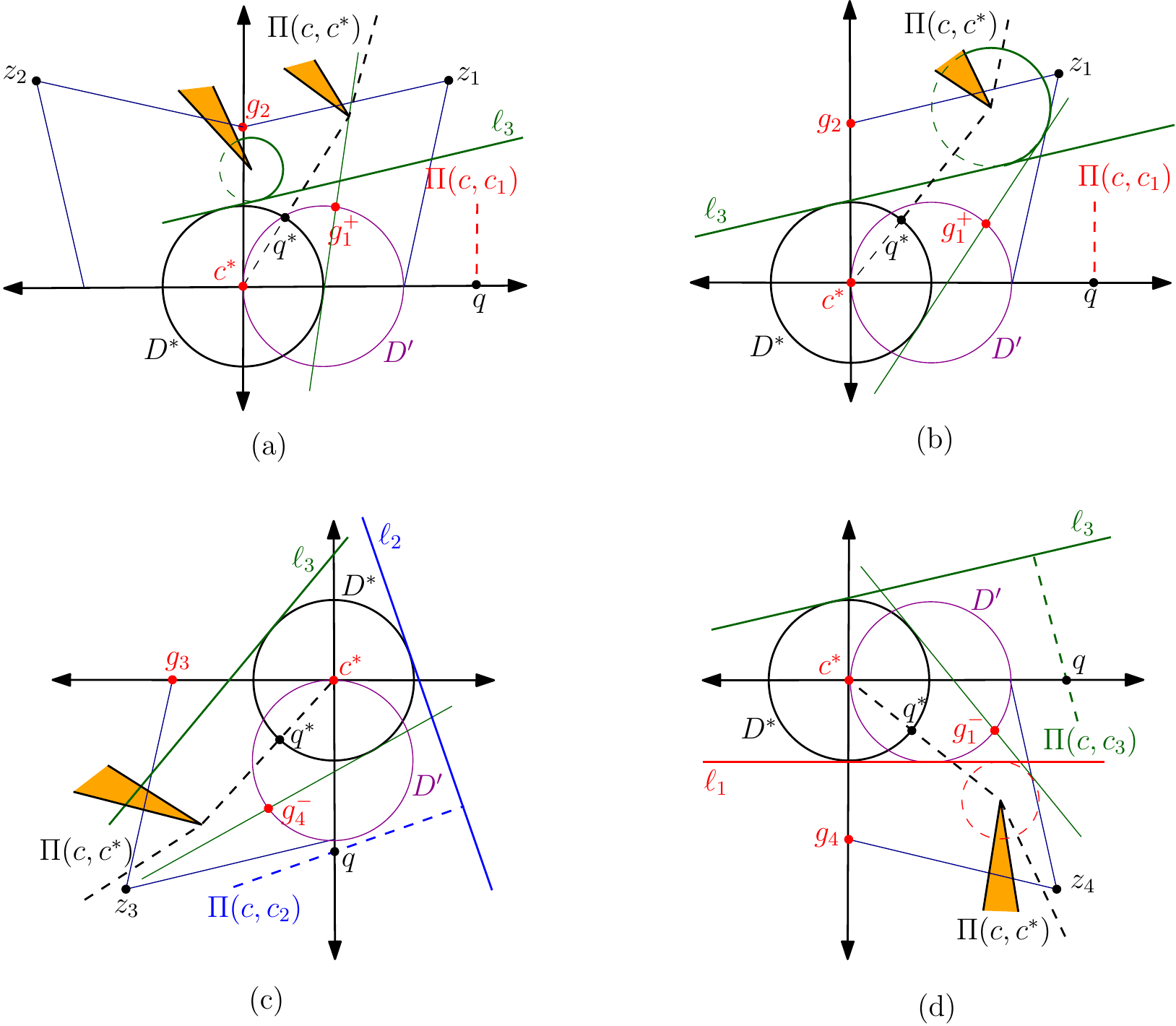}
\caption{Illustration of Lemma~\ref{lemma:g1+ right z1g2}: (a) Item (i), (b) Item (ii), (c) Item (iii), and (d) Item (iv).}
\label{fig:Lemma6}
\end{figure}
\begin{proof}
We prove Items (i) and (ii), the proofs of the other two items are symmetric to the proof of Item (ii).

Let $\l_v$ be the vertical line passing through $g^+_1$. Since $\Pi(c,c_1)$ intersects the $x$-axis at a point $q$ with $x(q) > x(g^+_1)$, $\Pi(c,c^*)$ intersects $\l_v$ at a point $p$ and the polygon cannot intersect the segment $\overline{pg^+_1}$.
Let $b=(b_x,2)$ be the point on $\Pi(c,c^*)$; see Figure~\ref{fig:y=2-1}. We distinguish between two cases. \\
{\bf Case~1: } $b_x \le \frac{3}{2}$.\\
{\bf Proof of item (i):} 
Since the polygon intersects the segment $\overline{z_2g_2}$, we have $y(q^*) \le y(g^+_1)$. Thus, 
the angle $\angle(p,g^+_1,q^*)$ is the largest in the triangle $\triangle(p,g^+_1,q^*)$; see Figure~\ref{fig:y=2-2}. Thus, $|pg^+_1| \le |\Pi(p,q^*)|$. Therefore, $|\Pi(c,g^+_1)| \le |\Pi(c,p)| + |pg^+_1| \le |\Pi(c,p)| + |\Pi(,q^*)| =  |\Pi(c,q^*)| \le r$. \\
\begin{figure}[htb]
    \centering
    \includegraphics[width=0.45\textwidth]{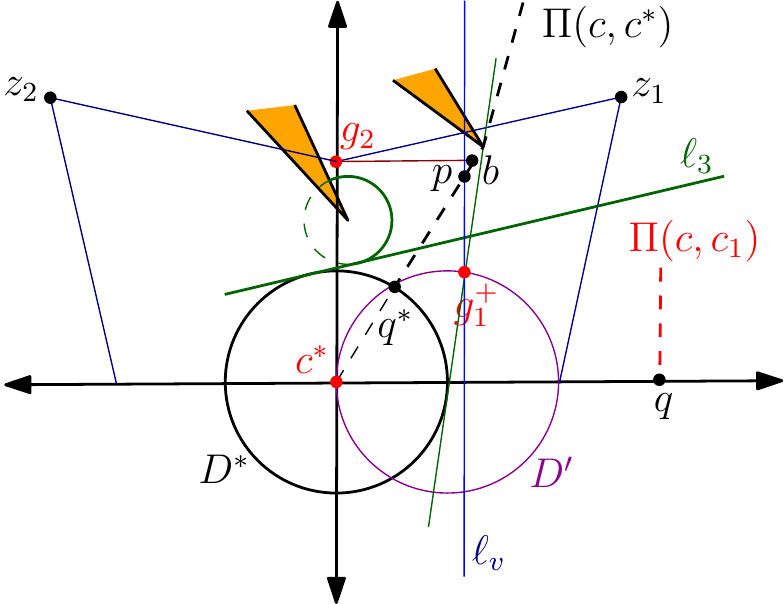}
    \caption{$b_x \le \frac{3}{2}$ and the polygon intersects the segment $\overline{z_2g_2}$.}
    \label{fig:y=2-2}
    \end{figure}
{\bf Proof of item (ii):}
\begin{itemize}
    \item If the polygon does not intersect the segment $\overline{g_2b}$, then, by Observation~\ref{obs:prabola1}, $|bg_2| \le |bq^*|$, and thus $|\Pi(c,g_2)| = |\Pi(c,b)| + |bg_2|\leq |\Pi(c,b)| + |bq^*| \le |\Pi(c,q^*)| \le r$; see Figure~\ref{fig:y=2-1}(a).
    \item Otherwise, the polygon intersects the segment $\overline{g_2b}$.
    \begin{itemize}
        \item If the polygon intersects the disk $D'$, then $g^+_1$ is defined as the intersection of the sweeping horizontal line $\l_h$ with $D'$, and thus $y(q^*) \le y(g^+_1)$; see Figure~\ref{fig:y=2-1}(b). 
        Thus, $|pg^+_1| \le |pq^*|$. 
        Therefore, since the polygon does not intersect the segment $\overline{pg^+_1}$, we have $|\Pi(c,g^+_1)| \le |\Pi(c,p)| + |pg^+_1| \le |\Pi(c,p)| + |pq^*| \le |\Pi(c,p)| + |\Pi(p,q^*)| =  |\Pi(c,q^*)| \le r$.
        \begin{figure}[H]
        \centering
        \includegraphics[width=0.88\textwidth]{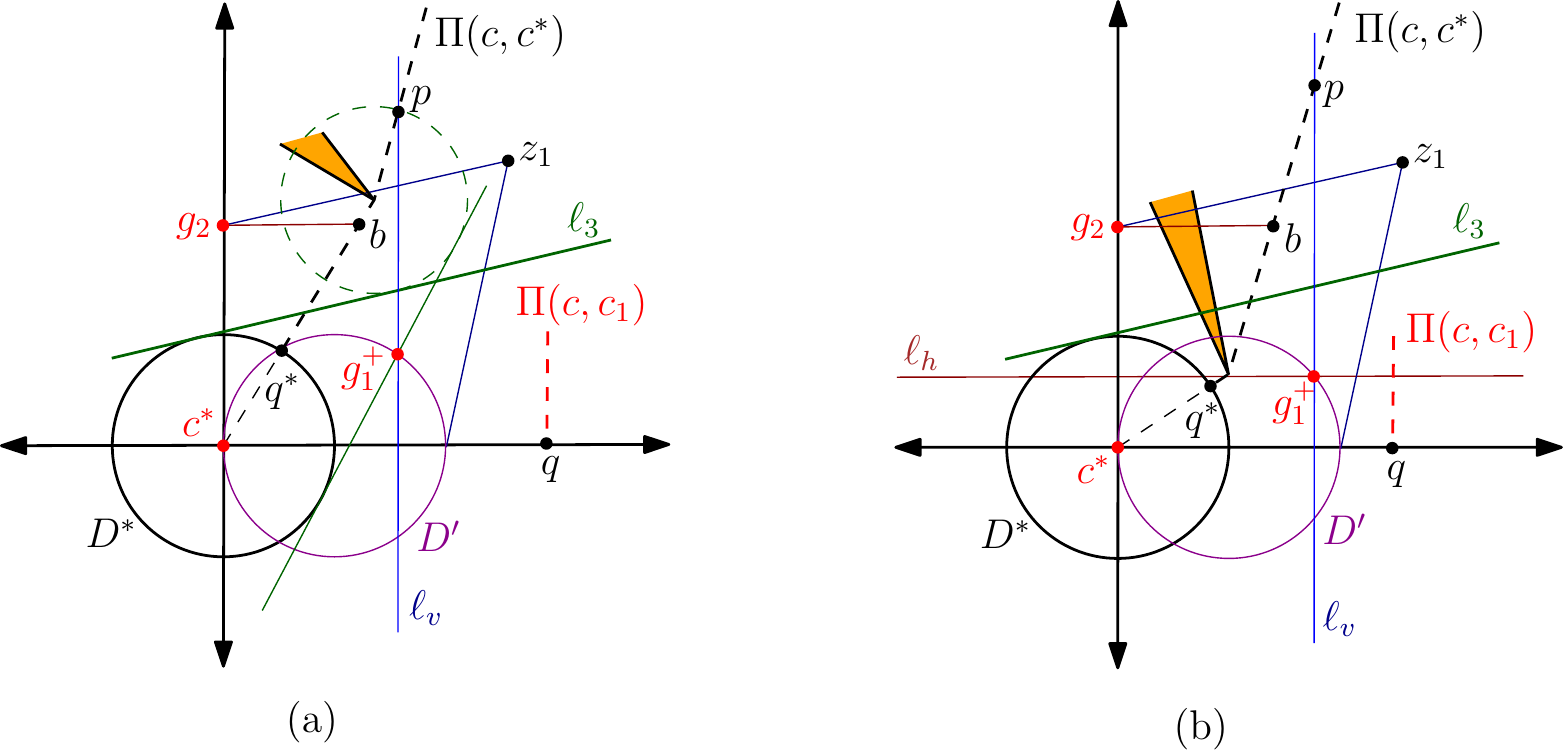}
        \caption{$b_x \le \frac{3}{2}$: (a) The polygon does not intersect the segment $\overline{g_2b}$. (b) The polygon intersects the disk $D'$.}
        \label{fig:y=2-1}
        \end{figure}
        \item Otherwise, let $\l'_3$ be the horizontal line that is tangent to $D^*$ at the point $(0,1)$ and let $D_b$ be the disk centered at $b$ and is tangent to $\l'_3$.
        Let $\l_b$ be a tangent line of $D_b$ and $D^*$ as depicted in Figure~\ref{fig:y=2}(b). 
        Let $g$ be the intersection point of $\l_b$ with $D'$ in $Q_1$.
        Since $g$ is below $g^+_1$ on the boundary of $D'$, we have $|\Pi(c,g^+_1)| \le |\Pi(c,g)|$.
        Therefore, to prove that $|\Pi(c,g^+_1)| \le r$, it is sufficient to prove that $|\Pi(c,g)|\leq r$.
        \begin{figure}[htb]
        \centering
        \includegraphics[width=0.88\textwidth]{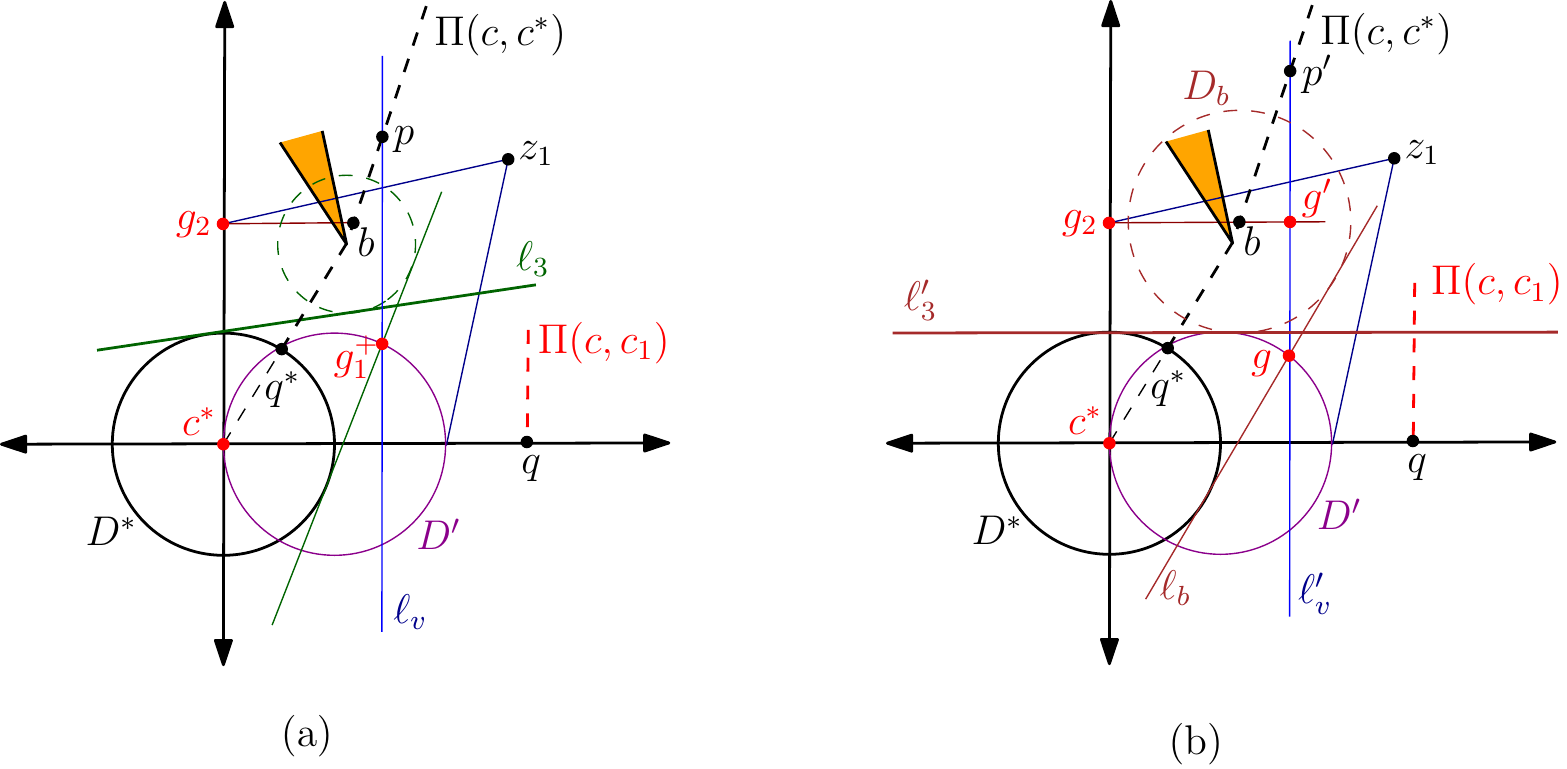}
        \caption{$b_x \le \frac{3}{2}$ and the polygon intersects the segment $\overline{g_2b}$ but not $D'$.}
        \label{fig:y=2}
        \end{figure}

        Let $\l'_v$ be the vertical line passing through $g$. Since $\Pi(c,c_1)$ intersects the $x$-axis at a point $q$ with $x(q) > x(z_1)$, $\Pi(c,c^*)$ intersects $\l'_v$ at a point $p'$ and the polygon cannot intersect the segment $\overline{p'g}$.
        Let $b_x$ denote the $x$-coordinate of $b$, and notice that the coordinates of $g =(g_x,g_y)$ depend on $b_x$.
        To compute the coordinates of $g$, we compute the intersection point between the tangent line $\l_b$ with $D'$ in $Q_1$.
        The equation of $\l_b$ is $y= \frac{2}{b_x}x  - \frac{\sqrt{4+b_x^2}}{b_x}$, and the equation of $D'$ is $(x-1)^2 + y^2 =1$.
        Hence, we have
        $$g_x = \frac{ b_x^2 + 2 \sqrt{4 + b_x^2} + 2 b_x \sqrt{\sqrt{4 + b_x^2} -1} }{4 + b_x^2} \ \text{ and } \ g_y = \frac{ 2 b_x - b_x \sqrt{4 + b_x^2} + 4  \sqrt{  \sqrt{4 + b_x^2} -1 } }{4 + b_x^2}. $$
        Since $\frac{ b_x^2 + 2 \sqrt{4 + b_x^2} + 2 b_x \sqrt{\sqrt{4 + b_x^2} -1} }{4 + b_x^2} - b_x > 0$, for every $0 < b_x \le \frac{3}{2}$, we have $b_x  \le g_x$, for every $0 < b_x \le \frac{3}{2}$. 

        Let $g'= (g_x,2)$ and notice that the polygon does not intersect the segment $\overline{g'g}$; see Figure~\ref{fig:y=2}(b).
        Moreover, since the polygon does not intersect the segment $\overline{p'g}$, we have $|\Pi(c,g')| \le |\Pi(c,b)|$.
        We now claim that $$ |g'g| =  2 - g_y < \sqrt{4+b^2_x} - 1 = |bq^*|, \text{ for every } 0 <b_x \le g_x.$$ 
        That is, $3 - \sqrt{4+b^2_x} <   g_y = \frac{ 2 b_x - b_x \sqrt{4 + b_x^2} + 4  \sqrt{  \sqrt{4 + b_x^2} -1 } }{4 + b_x^2}$.  
        To see the correctness of this inequality, we need to show that 
        $(3 - \sqrt{4+b^2_x})(4 + b_x^2) <  2 b_x - b_x \sqrt{4 + b_x^2} + 4  \sqrt{  \sqrt{4 + b_x^2} -1 }$. 
        This is true since the left side of this inequality has maximum value equals 4, when $b_x=0$, and the right side of this inequality has minimum value equals 4, when $b_x=0$, for each $0 < b_x \le \frac{3}{2}$.
        Therefore, we have $|\Pi(c,g)| \leq |\Pi(c,g')| + |g'g| \le |\Pi(c,b)| + |bq^*| \le |\Pi(c,q^*)| \le r$.
    \end{itemize}
\end{itemize}

\noindent
{\bf Case~2:} $b_x > \frac{3}{2}$. We show that in this case we have  $|\Pi(c,g^+_1)| \le r$, which proves both Items (i) and (ii).
Let $a = (a_x,a_y)$ be the intersection point of $\Pi(c,c^*)$ with the segment $\overline{z_1g_2}$. 
Since $\Pi(p,c^*)$ intersects $\l_v$, we have $a_x \ge b_x > \frac{3}{2}$.
Since the polygon does not intersect the segment $\overline{ag^+_1}$, we have $|\Pi(c, g^+_1)| \le |\Pi(c, a)| + |ag^+_1|$. Thus, it is sufficient to prove that $|ag^+_1|\leq |aq^*|$, for each $\frac{3}{2} \le a_x \le x(z_1)$.

Let $\l'_3$ be the horizontal line passing through point $(0,1)$ and let $D_a$ be the disk centered at $a$ and is tangent to $\l'_3$.
Let $\l_a$ be a tangent line to $D_a$ and $D^*$ as depicted in Figure~\ref{fig:z1Distance}(b). 
Let $g$ be the intersection point of $\l_a$ with $D'$ in $Q_1$. We distinguish between two cases. 
    \begin{figure}[htb]
    \centering
    \includegraphics[width=0.87\textwidth]{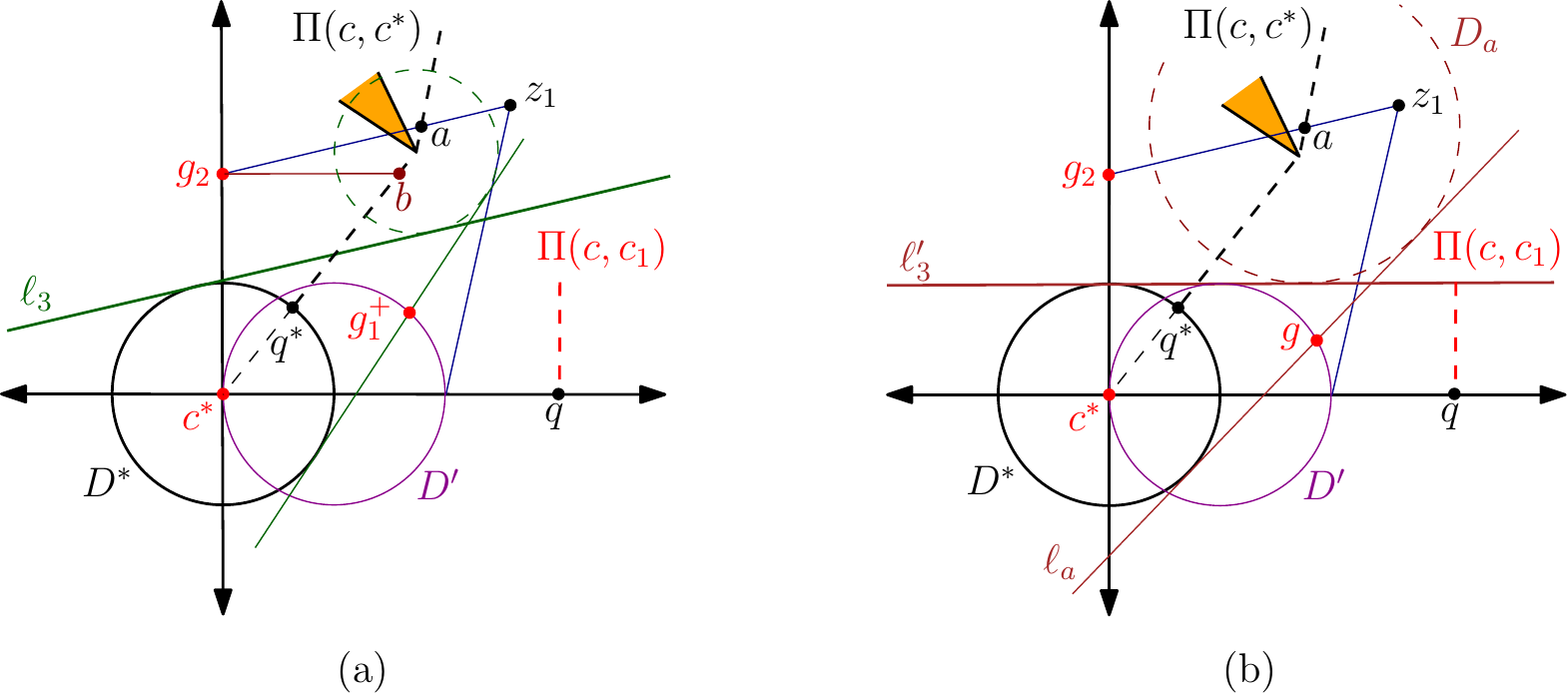}
    \caption{$b_x > \frac{3}{2}$ and $y(g) \le y(g^+_1)$.}
    \label{fig:z1Distance}
    \end{figure}

\noindent
{\bf Case~2.1:} $y(g) \le y(g^+_1)$ (i.e., $g$ is below $g^+_1$ on the boundary of $D'$). Since $1 \le x(g^+_1) \le 2$,  we have $|ag^+_1| \le |ag|$.
Therefore, to prove the lemma, it is sufficient to prove that $|ag|\leq |aq^*|$, for each $\frac{3}{2} \le a_x \le x(z_1)$.

Since $a$ is on the segment $\overline{z_1g_2}$ and the equation of the line passing through $z_1$ and $g_2$ is $y= \frac{4\sqrt{2} - 5}{3}x +2$, we have $a_y=\frac{(4\sqrt{2} - 5)a_x}{3} + 2$.
Notice that $x(g)$ and $y(g)$ (the coordinates of $g$) depend on $a_x$.
To compute these coordinates, we compute the intersection point between the tangent line $\l_a$ with $D'$ in $Q_1$.
The equation of $\l_a$ is $y= \left(x + \frac{3}{4\sqrt{2} - 5} \right)m - \frac{6}{(4\sqrt{2} - 5)a_x} - 1$, where
\begin{align*}
m = \frac{1}{8\big(5\sqrt{2} - 6\big)a_x} & \bigg( \big(12\sqrt{2} - 15\big)a_x + 18 \, - \\ 
& \ \ \ \sqrt{-3\Big( \big(120\sqrt{2} - 171\big)a^2_x - \big(1712\sqrt{2} - 2420\big)a_x + 480\sqrt{2} - 684\Big)} \ \bigg).
\end{align*}
Let $t$ be the point on the segment $\overline{z_1g_2}$, where $x(t) = \frac{3}{2}$. 
We prove that $|ag|\leq |aq^*|$, for each $ \frac{3}{2} \le a_x \le x(z_1)$ by dividing the segment $\overline{tz_1}$ into 7 intervals defined by the points $t_0,t_1,\dots,t_7$, where $x(t_0)=x(t)=\frac{3}{2}$, $x(t_1) = 1.52$, $x(t_2) = 1.56$, $x(t_3) = 1.63$, $x(t_4) = 1.74$, $x(t_5) = 1.9$, $x(t_6)= 2.15$, and $x(t_7) = x(z_1)=\frac{3(2+\sqrt{2})}{4}$. 
For each $1 \le i \le 7$, we compute the intersection point $g'_i$ of $l_a$ with the disk $D'$, where $a=t_i$, and we show that $|ag'_i| = \sqrt{(a_x - x(g'_i))^2 + (a_y - y(g'_i))^2} \le \sqrt{a^2_x + a^2_y} - 1 = |aq^*|$, for each $x(t_{i-1}) \le a_x \le x(t_{i})$.
\begin{itemize}
    \item For $i=1$, we have $g'_1 = (1.8033 , 0.5955)$, and thus $$\sqrt{(a_x - 1.8033)^2 + (a_y - 0.5955)^2} \le \sqrt{a^2_x + a^2_y} - 1, \ \text{ for each } 1.5 \le a_x \le 1.52.$$

    \item For $i=2$, we have $g'_2 = (1.8152, 0.5792)$, and thus $$\sqrt{(a_x - 1.8152)^2 + (a_y - 0.5792)^2} \le \sqrt{a^2_x + a^2_y} - 1, \ \text{ for each } 1.52 \le a_x \le 1.56.$$

    \item For $i=3$, we have $g'_3 = (1.8347 , 0.5507)$, and thus $$\sqrt{(a_x - 1.8347)^2 + (a_y - 0.5507)^2} \le \sqrt{a^2_x + a^2_y} - 1, \ \text{ for each } 1.56 \le a_x \le 1.63.$$

    \item For $i=4$, we have $g'_4 = (1.8623 , 0.5063)$, and thus $$\sqrt{(a_x - 1.8623)^2 + (a_y - 0.5063)^2} \le \sqrt{a^2_x + a^2_y} - 1, \ \text{ for each } 1.63 \le a_x \le 1.74.$$

    \item For $i=5$, we have $g'_5 = (1.8966 , 0.4429)$, and thus $$\sqrt{(a_x - 1.8966)^2 + (a_y - 0.4429)^2} \le \sqrt{a^2_x + a^2_y} - 1, \ \text{ for each } 1.174 \le a_x \le 1.9.$$

    \item For $i=6$, we have $g'_6 = (1.9376 , 0.3478)$, and thus $$\sqrt{(a_x - 1.9376)^2 + (a_y - 0.3478)^2} \le \sqrt{a^2_x + a^2_y} - 1, \ \text{ for each } 1.9 \le a_x \le 2.15.$$

    \item For $i=7$, we have $g'_7 = (1.9787 , 0.2053)$, and thus $$\sqrt{(a_x - 1.9787)^2 + (a_y - 0.2053)^2} \le \sqrt{a^2_x + a^2_y} - 1, \ \text{ for each } 2.15 \le a_x \le \frac{3(2+\sqrt{2})}{4}.$$
\end{itemize}
These inequalities hold since the function $\sqrt{a^2_x + a^2_y} - 1 - \sqrt{(a_x - x(g'_i))^2 + (a_y - y(g'_i))^2}$ is monotonic in the interval $x(t_{i-1}) \le a_x \le x(t_i)$, and has minimum value when $a_x = x(t_{i-1})$, for each $1 \le i \le 7$.
Thus, for each $a_x \le x(t_i)$, where $1 \le i \le 7$, we have $|ag| \le |ag'_i|$. This proves that $|ag| \le |aq^*|$.
    \begin{figure}[htb]
    \centering
    \includegraphics[width=0.9\textwidth]{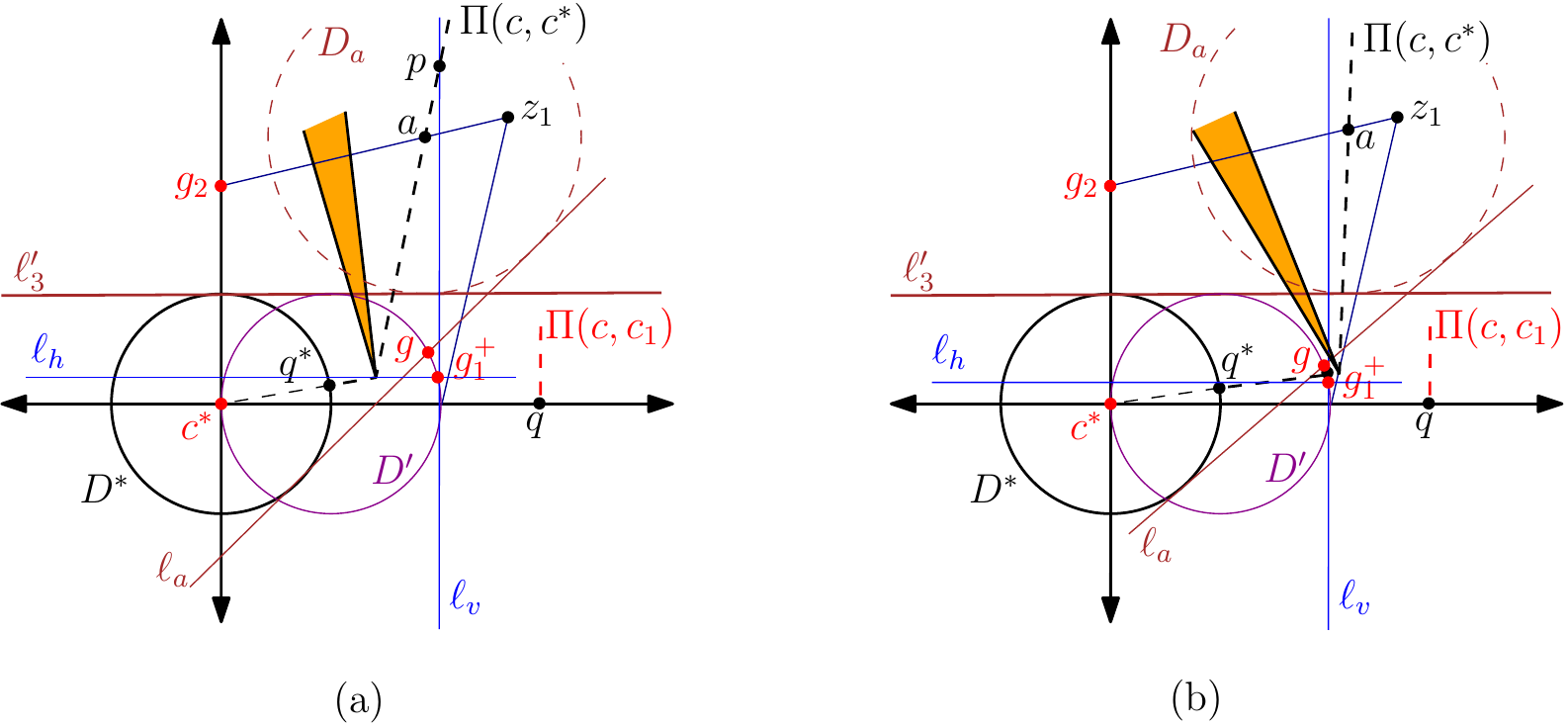}
    \caption{$b_x > \frac{3}{2}$ and $y(g') > y(g^+_1)$.}
    \label{fig:z1Distance2.2}
    \end{figure}

\noindent
{\bf Case~2.2:} $y(g) > y(g^+_1)$ (i.e., $g$ is above $g^+_1$ on the boundary of $D'$).
Observe that this case can happen only if the polygon intersects the line $\ell_3'$.
Recall that $\ell_{v}$ is the vertical line passing through $g^+_1$, $p$ is the intersection point of $\Pi(c,c^*)$ with $\l_v$, and the polygon does not intersect the segment $\overline{pg^+_1}$.
Let $\ell_{h}$ be the horizontal line passing through $g^+_1$.
Since the polygon intersects $\ell_3'$, we have $y(q^*) \le y(g^+_1)$, i.e., $q^*$ is below $\ell_{h}$; see Figure~\ref{fig:z1Distance2.2}.
Thus, the angle $\angle (p,g^+_1,q^*) \geq \frac{\pi}{2}$, and, since the polygon does not intersect $\overline{pg^+_1}$, we have $|pg^+_1| \le |\Pi(p,q^*)|$.
Therefore, $|\Pi(c,g^+_1)|\leq |\Pi(c,p)|+|pg^+_1| \le |\Pi(c,p)| + |\Pi(p,q^*)| = |\Pi(c,q^*)| \le r$.
\end{proof}

\begin{lemma}\label{lemma:g1+ right z1g1}
Let $D\in \D$ be a disk centered at $c$ with radius $r$.
\begin{enumerate}
    \item[(i)] If $c' \in Q_1$, $\Pi(c,c_1)$ intersects the $x$-axis at a point $q$ with $x(q) > x(g^+_1)$, and $\Pi(c,c^*)$ intersects the segment $\overline{z_1g_1}$, then $|\Pi(c,g^+_1)|\leq r$; see Figure~\ref{fig:Lemma7}(a).
   \item[(ii)] if $c' \in Q_3$, $\Pi(c,c_2)$ intersects the $y$-axis at a point $q$ with $y(q) < x(g^-_4)$, and $\Pi(c,c^*)$ intersects the segment $\overline{z_3g_4}$, then $|\Pi(c,g^-_4)|\leq r$; see Figure~\ref{fig:Lemma7}(b).
   \item[(iii)] If $c' \in Q_4$, $\Pi(c,c_3)$ intersects the $x$-axis at a point $q$ with $x(q) > x(g^-_1)$, and $\Pi(c,c^*)$ intersects the segment $\overline{z_4g_1}$, then $|\Pi(c,g^-_1)|\leq r$; see Figure~\ref{fig:Lemma7}(c).
\end{enumerate}
\end{lemma}
\begin{figure}[H]
\centering
\includegraphics[width=0.83\textwidth]{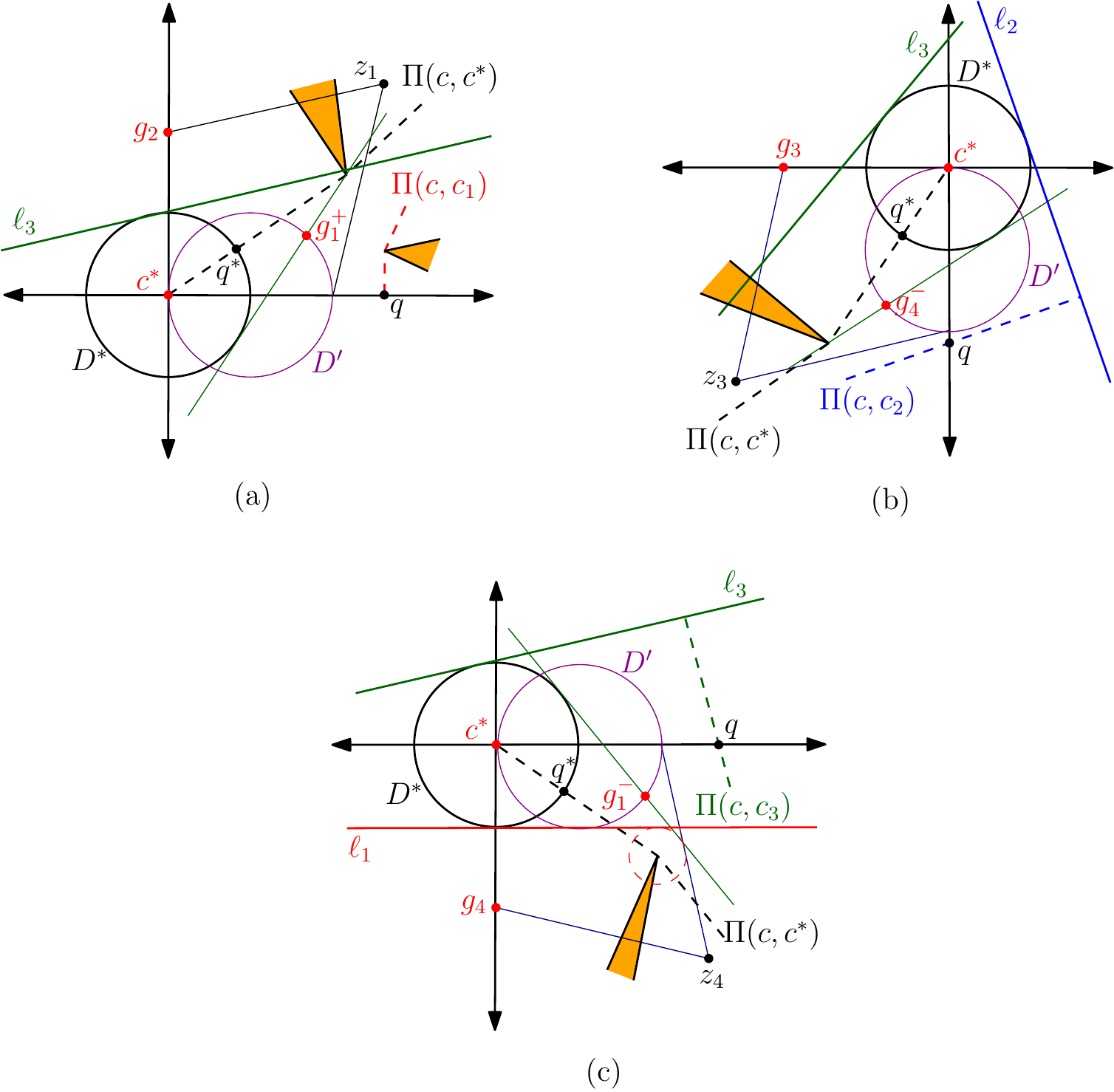}
\caption{Illustration of Lemma~\ref{lemma:g1+ right z1g1}: Item (i), (b) Item (ii), and (c) Item (iii).}
\label{fig:Lemma7}
\end{figure}
\begin{proof}
We prove Item (i), the proofs of the other items are symmetric.

Let $\l_v$ be the vertical line passing through $g^+_1$, and let $p$ be the intersection point of $\Pi(c,c^*)$ with $\l_v$; see Figure~\ref{fig:g1+}.
We distinguish between two cases.\\
{\bf Case~1:} $y(p) \ge y(g^+_1)$; see Figure~\ref{fig:g1+}(a).
Let $a$ be the intersection point of $\Pi(c,c^*)$ with $\overline{z_1g_1}$. 
By the definition of $z_1$, we have $|ag_1| \le |aq^*|$, and thus $|\Pi(c,g_1)|\leq r$.
Moreover, since $p$ is above $g^+_1$, $g^+_1$ is inside the pseudo-triangle $\triangle(c,q^*,g_1)$; see Figure~\ref{fig:g1+}(a).
Thus, by Observation~\ref{obs:observationTriangle}, $D$ contains $g^+_1$, and therefore $|\Pi(c,g^+_1)|\le r$. \\
\begin{figure}[htb]
\centering
\includegraphics[width=0.9\textwidth]{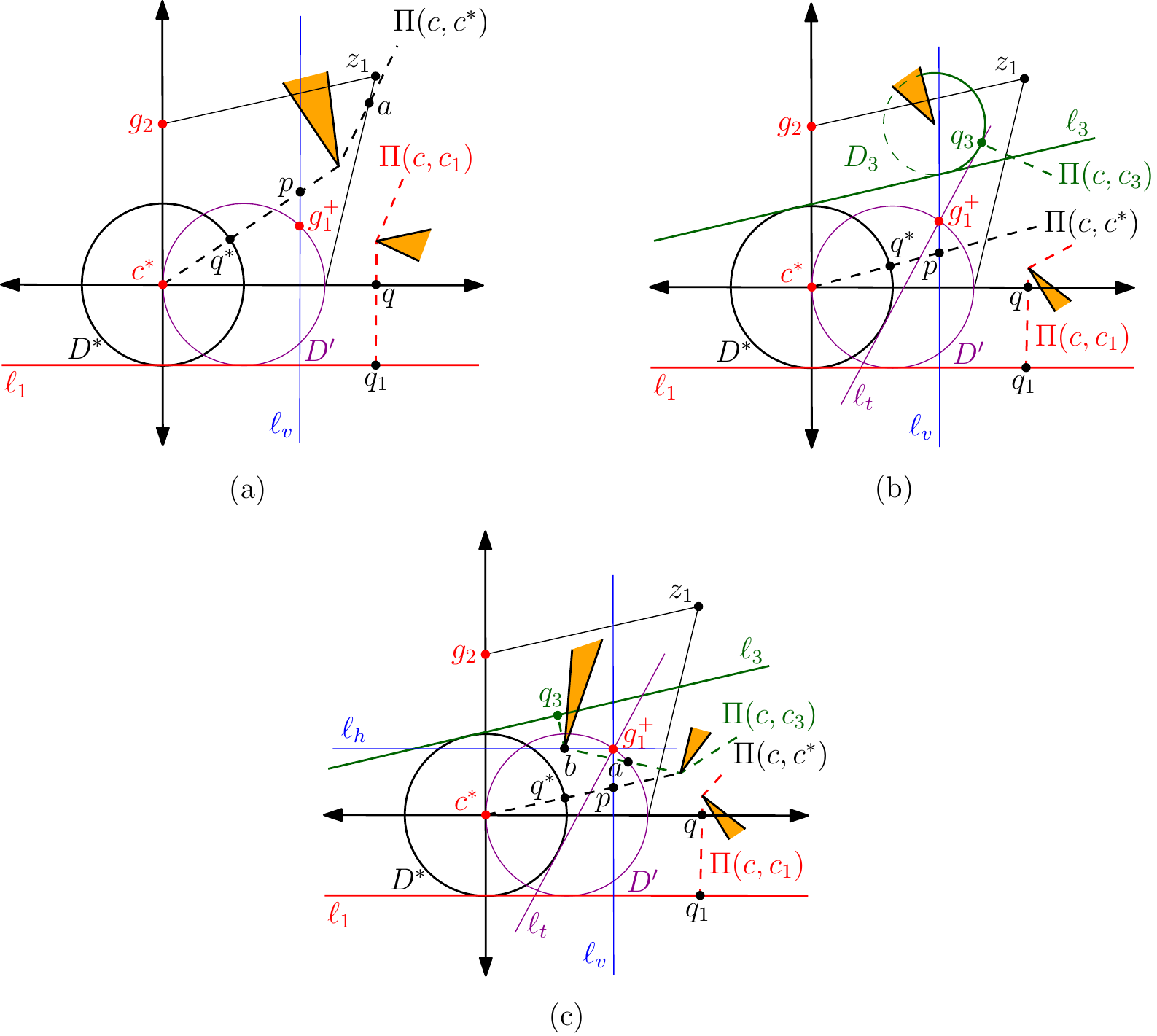}
\caption{Illustration of the proof of Lemma~\ref{lemma:g1+ right z1g1}, Item (i): (a) $y(p) \ge y(g^+_1)$, (b) $y(p) < y(g^+_1)$ and $\Pi(c,c_3)$ intersects $\l_t$ above $g^+_1$, and (c) $y(p) < y(g^+_1)$ and $\Pi(c,c_3)$ intersects $\l_t$ below $g^+_1$.}
\label{fig:g1+}
\end{figure}

\noindent
{\bf Case~2:} $y(p) < y(g^+_1)$; see Figure~\ref{fig:g1+}(b).
Let $\l_t$ be the line that is tangent to $D^*$ and passes through $g^+_1$, and observe that $\Pi(c,c_3)$ intersects this line.
\begin{itemize}
    \item If $\Pi(c,c_3)$ intersects $\l_t$ above $g^+_1$, then $g^+_1$ is inside the pseudo-triangle $\triangle(c,q^*,q_3)$; see Figure~\ref{fig:g1+}(b). 
    Thus, by Observation~\ref{obs:observationTriangle}, $D$ contains $jg^+_1$, and therefore $|\Pi(c,g^+_1)|\le r$.
    \item If $\Pi(c,c_3)$ intersects $\l_t$ below $g^+_1$, then let $\l_h$ be the horizontal line passing through $g^+_1$; see Figure~\ref{fig:g1+}(c). Let $a$ be the intersection point of $\Pi(c,c_3)$ with the boundary of $D'$, and let $b$ be the intersection point of $\Pi(c,c_3)$ with $\l_h$. 
    Observe that $x(b) \le x(g^+_1) \le x(a)$ and $y(b) = y(g^+_1) \ge y(a)$.
    Hence, the angle $\angle(a,g^+_1,b)$ is the largest in the triangle $\triangle(a,g^+_1,b)$. Thus, $|ag^+_1| \le |ab| \le |\Pi(a,q_3)|$. Therefore, $|\Pi(c,g^+_1)| \le |\Pi(c,a)| + |ag^+_1| \le |\Pi(c,a)| + |\Pi(a,q_3)| = |\Pi(c,q_3)| \le r$.
\end{itemize}
\vspace{-0.7cm}
\end{proof}

\begin{lemma}\label{lemma:not intersecting z1g1 z1g2}
Let $D\in \D$ be a disk centered at $c$ with radius $r$.
\begin{enumerate}
    \item[(i)] If $c' \in Q_1$ and $\Pi(c,c^*)$ does not intersect the segments $\overline{z_1g_1}$ nor $\overline{z_1g_2}$, then $|\Pi(c,c^*)|\le r$ or $|\Pi(c,g^+_1)|\le r$;  see Figure~\ref{fig:Lemma8}(a).
    \item[(ii)] If $c' \in Q_3$, $\Pi(c,c^*)$ does not intersect the segments $\overline{z_3g_4}$ nor $\overline{z_3g_3}$, and $\Pi(c,c_2)$ intersects the negative $y$-axis, then $|\Pi(c,c^*)|\le r$ or $|\Pi(c,g^-_4)|\le r$; see Figure~\ref{fig:Lemma8}(b).
    \item[(iii)] If $c' \in Q_4$, $\Pi(c,c^*)$ does not intersect the segments $\overline{z_4g_4}$ nor $\overline{z_4g_1}$, and $\Pi(c,c_3)$ intersects the positive $x$-axis, then $|\Pi(c,c^*)|\le r$ or $|\Pi(c,g^-_1)|\le r$; see Figure~\ref{fig:Lemma8}(c).
\end{enumerate}
\end{lemma}
\begin{figure}[H]
\centering
\includegraphics[width=0.82\textwidth]{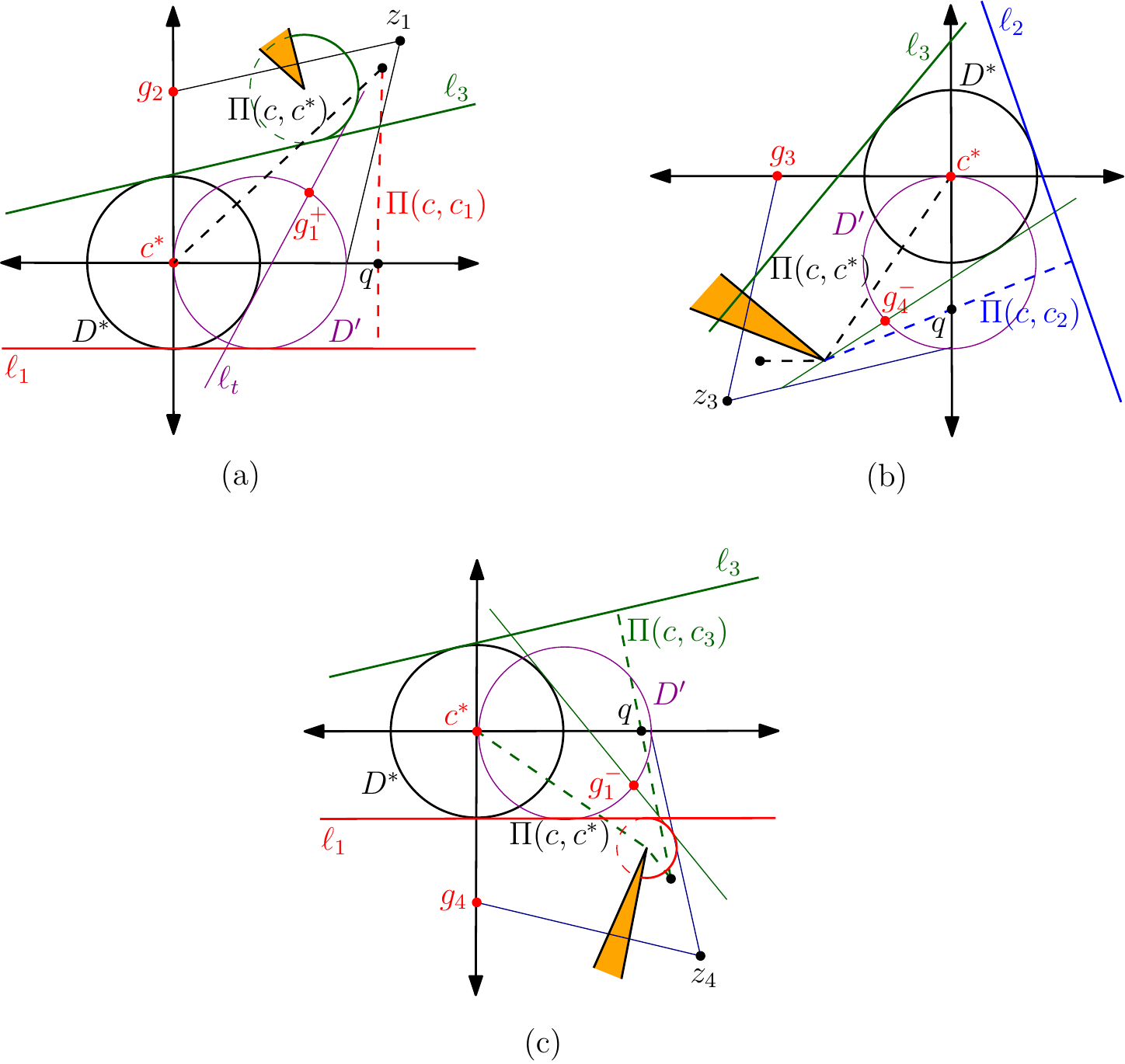}
\caption{Illustration of Lemma~\ref{lemma:not intersecting z1g1 z1g2}: (a) Item (i), (b) Item (ii), and (c) Item (iii).}
\label{fig:Lemma8}
\end{figure}
\begin{proof}
We prove Item (i), the proofs of the other two items are symmetric.

Since $c' \in Q_1$, the path $\Pi(c,c_1)$ intersects the positive $x$-axis at a point $q$.
If $x(c^*) \le x(q) \le x(g^+_1)$, then by Lemma~\ref{lemma:c*g'1}, $|\Pi(c,c^*)|\le r$ or $|\Pi(c,g^-_1)|\le r$.
Otherwise, $x(g^+_1) < x(q) \le x(z_1)$.
Let $\l_h$ be the horizontal line passing through $g^+_1$.
\begin{itemize}
    \item If $\Pi(c,c_1)$ intersects $\l_h$, then let $p$ be this intersection point; see Figure~\ref{fig:not intersecting z1g1 z1g2}(a). Thus, $x(g^+_1) \le x(p) \le x(z_1)$ and the polygon does not intersect the segment $\overline{g^+_1p}$.
    Let $q_1$ be the intersection point of $\Pi(c,c_1)$ with $\l_1$.
    Since, $|pg^+_1| = |pq| + (x(p) -2)$, $x(p) < 3$, and $|pq_1|= |pq| + 1$, we have $|pg^+_1| < |pq_1| \le |\Pi(p,q_1)|$. 
    Therefore, since $|\Pi(c,q_1)| \le r$, we have  $|\Pi(c,g^+_1)|\le |\Pi(c,p)| + |pg^+_1| \le |\Pi(c,p)| + |\Pi(p,q_1)| = |\Pi(c,q_1)| \le r$.
    \item If $\Pi(c,c_1)$ does not intersect $\l_h$, then let $\l_t$ be the tangent to $D^*$ with positive slope that passes through $g^+_1$; see Figure~\ref{fig:not intersecting z1g1 z1g2}(b).  
     \begin{itemize}
         \item \ If $\Pi(c,c_3)$ intersects $\l_t$ above $g^+_1$, then $g^+_1$ is inside the pseudo-triangle $\triangle(c,q^*,q_3)$. Thus, by Observation~\ref{obs:observationTriangle}, $D$ contains $g^+_1$, and therefore $|\Pi(c,g^+_1)|\le r$.
        \item \ If $\Pi(c,c_3)$ intersects $\l_t$ below $g^+_1$, then let $\l_h$ be the horizontal line passing through $g^+_1$. 
        Let $a$ be the intersection point of $\Pi(c,c_3)$ with the boundary of $D'$, and let $b$ be the intersection point of $\Pi(c,c_3)$ with $\l_h$. Observe that $x(b) \le x(g^+_1) \le x(a)$ and $y(b) = y(g^+_1) \ge y(a)$.
        Hence, the angle $\angle(a,g^+_1,b)$ is the largest in the triangle $\triangle(a,g^+_1,b)$. Thus, $|ag^+_1| \le |ab| \le |\Pi(a,q_3)|$. Therefore, $\Pi(c,g^+_1) \le |\Pi(c,a)| + |ag^+_1| \le |\Pi(c,a)| + |\Pi(a,q_3)| = |\Pi(c,q_3)| \le r$.
    \end{itemize}
\end{itemize}
\vspace{-0.7cm}
\end{proof}
\begin{figure}[htb]
\centering
\includegraphics[width=0.87\textwidth]{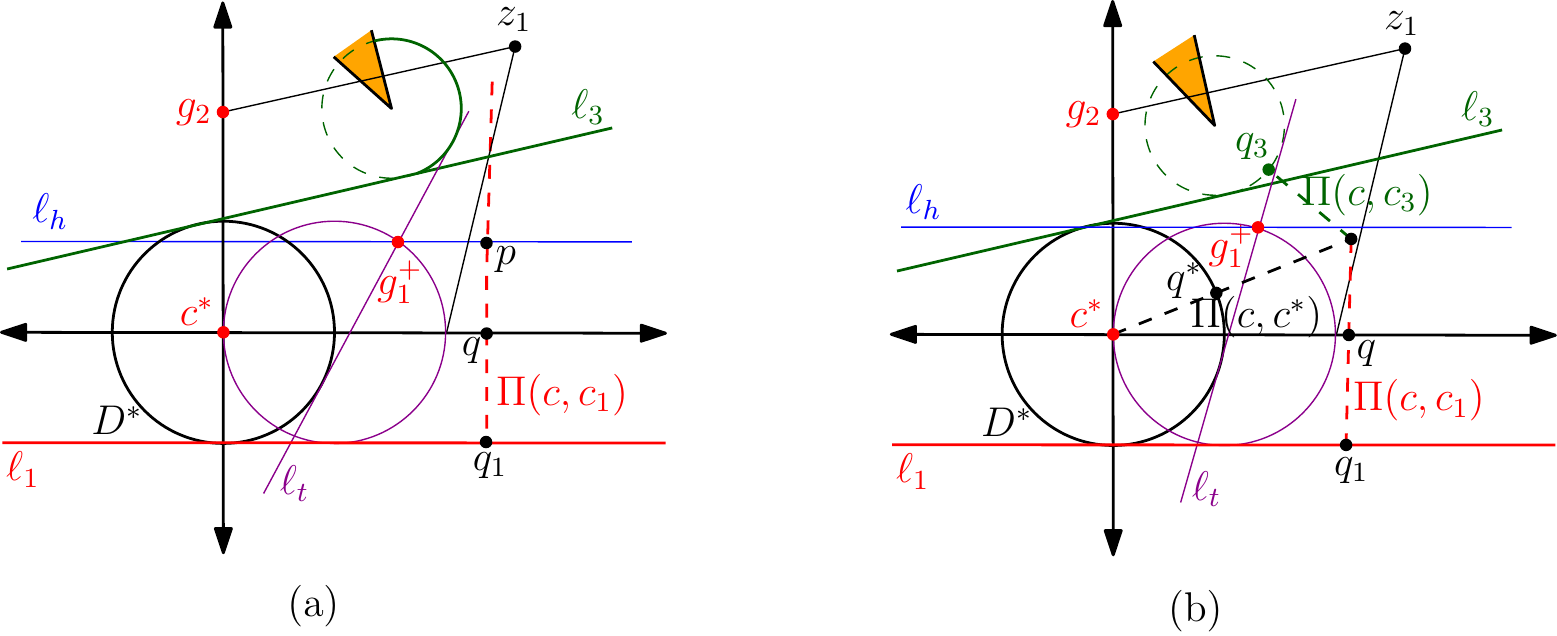}
\caption{Illustration of the proof of Lemma~\ref{lemma:not intersecting z1g1 z1g2}, Item (i): (a) $\Pi(c,c_1)$ intersects $\l_h$, and (b) $\Pi(c,c_1)$ does not intersect $\l_h$.}
\label{fig:not intersecting z1g1 z1g2}
\end{figure}

\begin{lemma}\label{lemma:g1-}
Let $D\in \D$ be a disk centered at $c$ with radius $r$.
\begin{enumerate}
    \item[(i)] If $c' \in Q_1$ and $g^- _1 \neq g_1$, 
    then $|\Pi(c,c^*)|\le r$ or $|\Pi(c,g^-_1)|\le r$; see Figure~\ref{fig:Lemma9}(a).
    \item[(ii)] If $c' \in Q_4$, $g^+_1 \neq g_1$ and $\Pi(c,c_3)$ intersects the positive $x$-axis, then $|\Pi(c,c^*)|\le r$ or $|\Pi(c,g^+_1)|\le r$; see Figure~\ref{fig:Lemma9}(b)
      \item[(iii)] If $c' \in Q_4$, $g^-_4 \neq g_4$ and $\Pi(c,c_3)$ intersects the negative $y$-axis, then $|\Pi(c,c^*)|\le r$ or $|\Pi(c,g^-_4)|\le r$; 
     see Figure~\ref{fig:Lemma9}(c). 
\end{enumerate}
\end{lemma}
\begin{figure}[htb]
\centering
\includegraphics[width=0.77\textwidth]{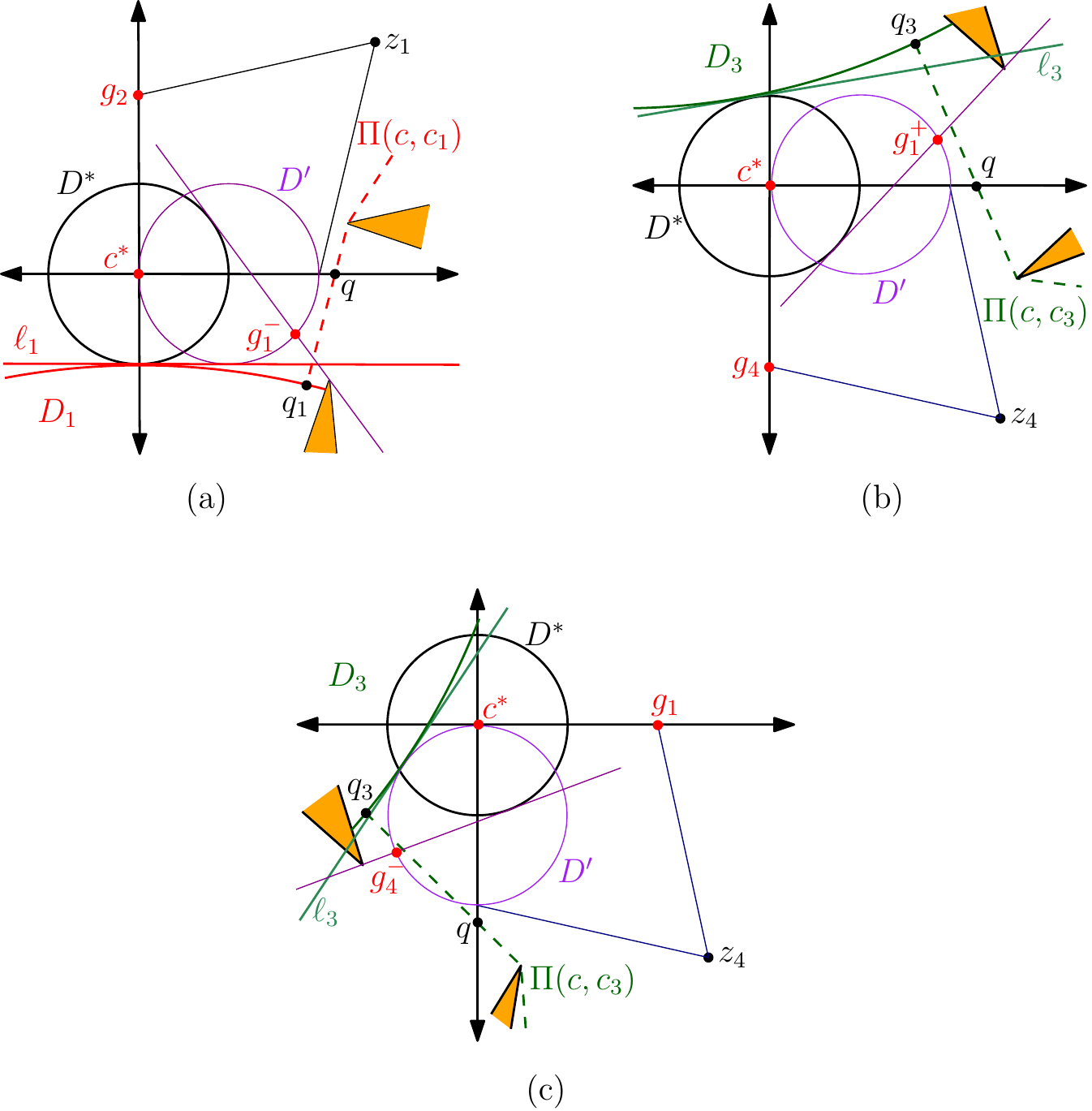}
\caption{Illustration of Lemma~\ref{lemma:g1-}: (a) Item (i), (b) Item (ii), and (c) Item (iii).}
\label{fig:Lemma9}
\end{figure}

\begin{proof}
We prove Item (i), the proofs of the other two items are symmetric.

Since $c' \in Q_1$, $\Pi(c,c_1)$ intersects the positive $x$-axis at a point $q$.
If $x(c^*) \le x(q) \le x(g^-_1)$, then, by Lemma~\ref{lemma:c*g'1}, $|\Pi(c,c^*)|\le r$ or $|\Pi(c,g^-_1)|\le r$.
Otherwise, $x(q) > x(g^-_1)$.
Let $\l_t$ be the line of negative slope that is tangent to $D^*$ and passes through $g^-_1$, and observe that if $g^- _1 \neq g_1$, then  $\Pi(c,c_1)$ intersects this line; see Figure~\ref{fig:g1-}.
\begin{itemize}
    \item If $\Pi(c,c_1)$ intersects $\l_t$ below $g^-_1$, then $g^-_1$ is inside the pseudo-triangle $\triangle(c,q^*,q_1)$; see Figure~\ref{fig:g1-}(a). 
    Thus, by Observation~\ref{obs:observationTriangle}, $D$ contains $g^-_1$, and therefore $|\Pi(c,g^-_1)|\le r$.
    \item If $\Pi(c,c_1)$ intersects $\l_t$ above $g^-_1$, then let $\l_h$ be the horizontal line passing through $g^-_1$; see Figure~\ref{fig:g1-}(b). Let $a$ be the intersection point of $\Pi(c,c_1)$ with the boundary of $D'$, and let $b$ be the intersection point of $\Pi(c,c_1)$ with $\l_h$. 
    Observe that $x(b) \le x(g^-_1) \le x(a)$ and $y(b) = y(g^-_1) \le y(a)$.
    Hence, the angle $\angle(a,g^-_1,b)$ is the largest in the triangle $\triangle(a,g^-_1,b)$. Thus, $|ag^-_1| \le |ab| \le |\Pi(a,q_1)|$. Therefore, $|\Pi(c,g^-_1)| \le |\Pi(c,a)| + |ag^-_1| \le |\Pi(c,a)| + |\Pi(a,q_1)| = |\Pi(c,q_1)| \le r$.
\end{itemize}
\vspace{-0.7cm}
\end{proof}
\begin{figure}[htb]
\centering
\includegraphics[width=0.88\textwidth]{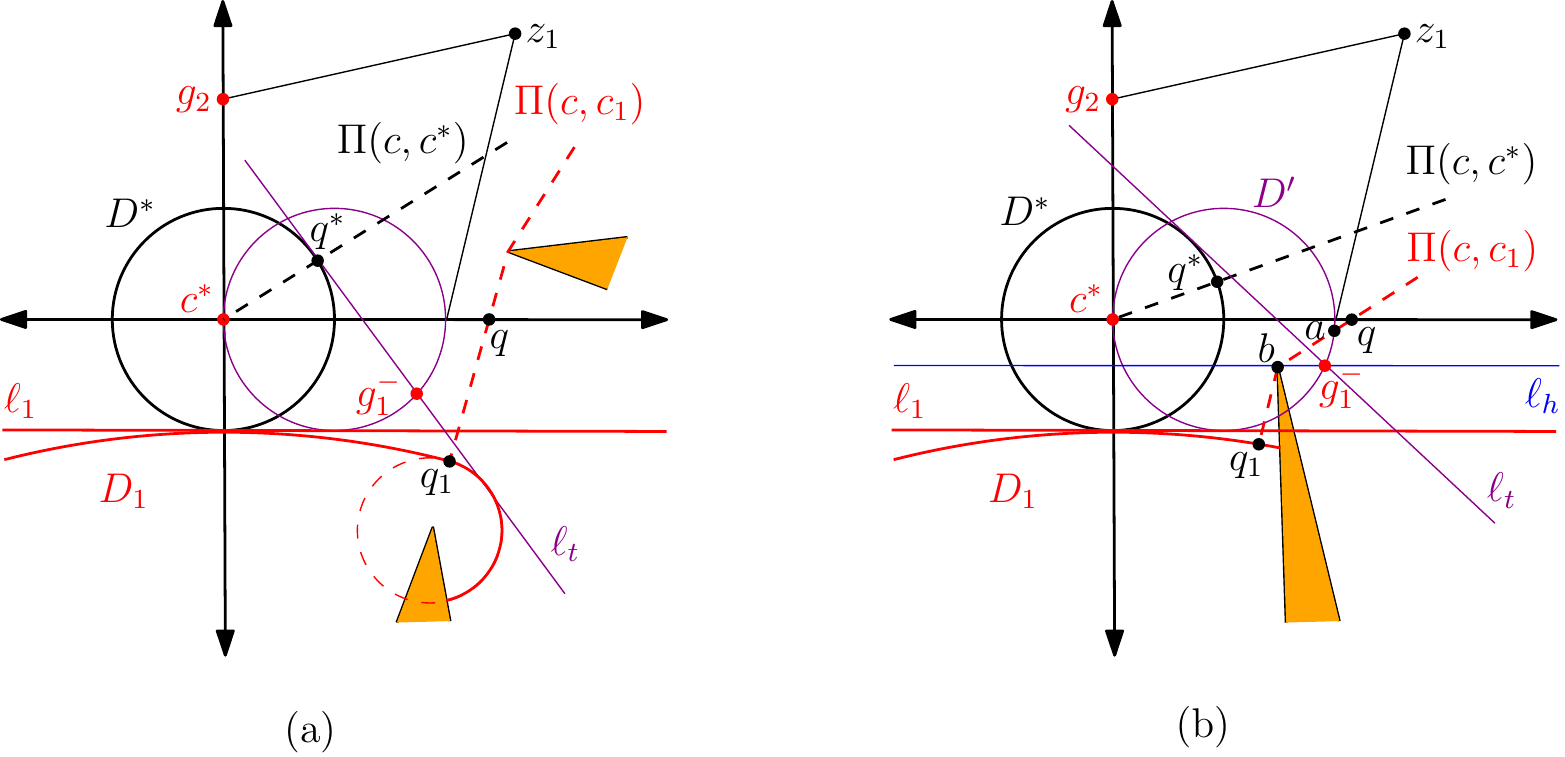}
\caption{Illustration of the proof of Lemma~\ref{lemma:g1-}, Item (i): (a) $\Pi(c,c_1)$ intersects $\l_t$ below $g^-_1$, and (b) $\Pi(c,c_1)$ intersects $\l_t$ above $g^-_1$.}
\label{fig:g1-}
\end{figure}

\begin{lemma}\label{lemma:g2+}
Let $D\in \D$ be a disk centered at $c$ with radius $r$ and $c' \in Q_2$, such that $\Pi(c,c_1)$ intersects the $x$-axis at a point $q$ with $x(q) < x(z_2)$, and $\Pi(c,c^*)$ intersects the segment $\overline{z_2g_2}$.
If $\alpha_2 > \frac{\pi}{3}$, then $|\Pi(c,g^+_2)| \le r$.
\end{lemma}
\begin{proof}
Let $\l_{v}$ be the vertical line passing through $g^+_2$ and let $p$ be the intersection point of $\Pi(c,c^*)$ with $\l_v$; see Figure~\ref{fig:g2+}.
\begin{itemize}
    \item If $y(p) \ge y(g^+_2) = 1$, then, since $y(q^*) \le  1 $, the angle $\angle(p,g^+_2,q^*)$ is the largest in the triangle $\triangle(p,g^+_2,q^*)$; see Figure~\ref{fig:g2+}(a). Since the polygon does not intersect $\overline{pg^+_2}$, we have $|pg^+_2| \le |\Pi(p,q^*)|$. Therefore, $|\Pi(c,g^+_2)| \le |\Pi(c,p)| + |pg^+_2| \le |\Pi(c,p)| + |\Pi(p,q^*)| = |\Pi(c,q^*)| \le r$.
    \item If $y(p) < y(g^+_2)$, then consider the path $\Pi(c,c_2)$ and notice that, since $\alpha_2 > \frac{\pi}{3}$, this path intersects $\l_v$ at a point $a$; see Figure~\ref{fig:g2+}(b). If $y(a) \ge y(g^+_2)$, then $g^+_2$ is inside the pseudo triangle $\triangle(c,q^*,q_2)$, and by Observation~\ref{obs:observationTriangle}, $D$ contains $g^+_2$. Otherwise, $0 \le y(a) < y(g^+_2)$. In this case, $|ag^+_2| \le 1$, and, since $\alpha_2 > \frac{\pi}{3}$, we have $|\Pi(a,q_2)| > 1$. Moreover, since the polygon does not intersect $\overline{ag^+_2}$, we have $|ag^+_2| < |\Pi(a,q_2)|$. Therefore, $|\Pi(c,g^+_2)| \le |\Pi(c,a)| + |ag^+_2| < |\Pi(c,a)| + |\Pi(a,q_2)| = |\Pi(c,q_2)| \le r$.
\end{itemize}
\vspace{-0.7cm}
\end{proof}
\begin{figure}[htb]
\centering
\includegraphics[width=0.83\textwidth]{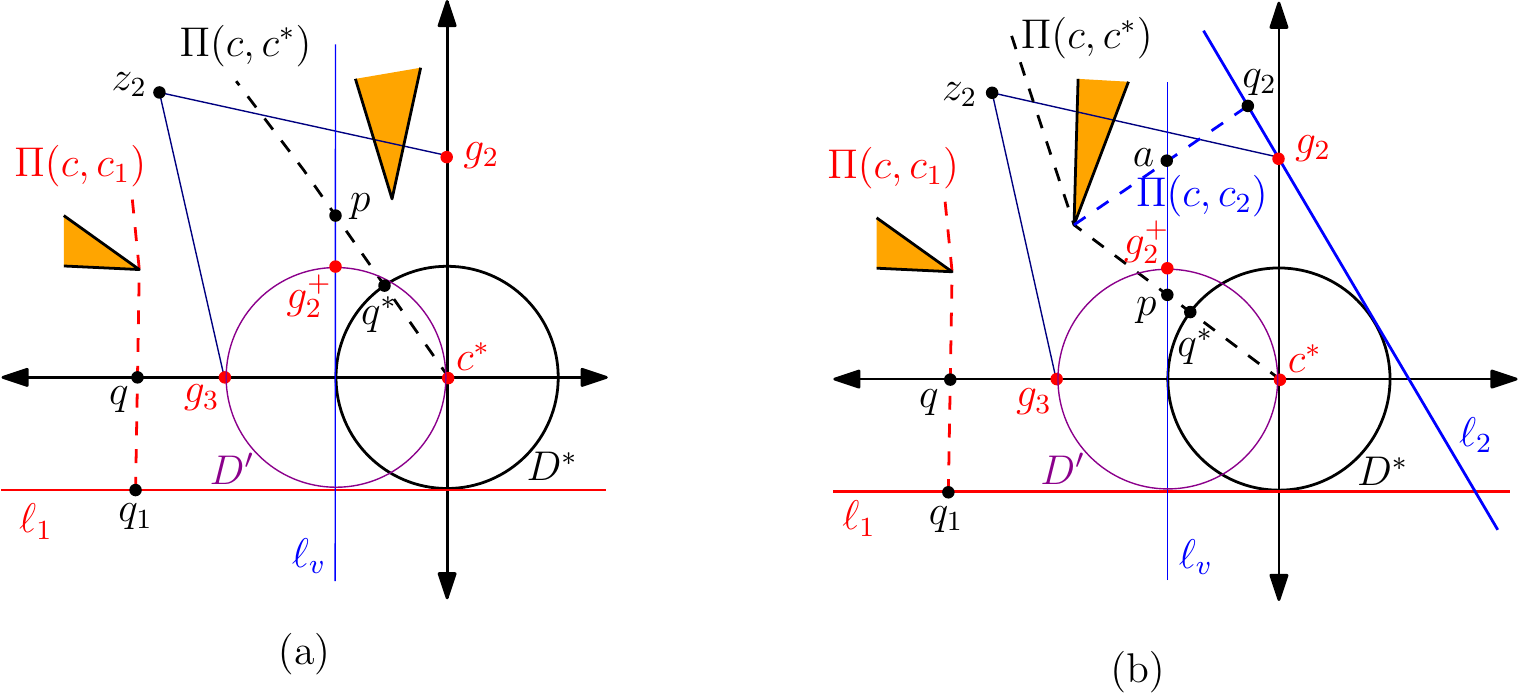}
\caption{Illustration of the proof of Lemma~\ref{lemma:g2+}: (a) $\Pi(c,c^*)$ intersects $\l_v$ above $g^+_2$, and (b) $\Pi(c,c^*)$ intersects $\l_v$ below $g^+_2$.}
\label{fig:g2+}
\end{figure}

\begin{lemma}\label{lemma:Q4g1}
Let $D\in \D$ be a disk centered at $c$ with radius $r$ and $c' \in Q_4$, such that $\Pi(c,c_3)$ intersects the $x$-axis at a point $q$ where $x(q)>x(g_1)$. 
If $\alpha_3 > \frac{\pi}{6}$, then 
\begin{itemize}
    \item if $\Pi(c,c^*)$ intersects $\overline{z_4g_1}$, then $|\Pi(c,g_1)|\le r$; and
    \item if $\Pi(c,c^*)$ intersects $\overline{z_4g_4}$, then $|\Pi(c,g_1)|\le r$ and $|\Pi(c,g_1^-)|\le r$.
\end{itemize}
\end{lemma}
\begin{proof}
Let $\l_t$ be the line tangent to $D^*$ with a positive slope that passes through $g_1$, and notice that the acute angle between $\l_t$ and the $x$-axis is $\frac{\pi}{3}$.
Since $\alpha_3 > \frac{\pi}{6}$, $\Pi(c,c_3)$ intersects $\l_t$.
Moreover, since $x(q)>x(g_1)$, $g_1$ is inside the pseudo-triangle $\triangle(c,q^*,q_3)$; see Figure~\ref{fig:Q4g1}(a). Thus, by Observation~\ref{obs:observationTriangle}, $D$ contains $g_1$, and therefore $|\Pi(c,g_1)|\le r$.
Let $g=(-1,-1)$, and notice that $g^-_1$ is on the small arc $\arc{g_1g}$ of $D'$ between $g_1$ and $g$; see Figure~\ref{fig:Q4g1}(b).
If $\Pi(c,c^*)$ intersects $\overline{z_4g_4}$, then $\arc{g_1g}$ is contained in the pseudo-triangle $\triangle(c,q^*,q_3)$. Thus, by Observation~\ref{obs:observationTriangle}, $D$ contains both $g_1$ and $g^-_1$. Therefore, $|\Pi(c,g_1)|\le r$ and $|\Pi(c,g^-_1)|\le r$.
\end{proof}
    \begin{figure}[htb]
    \centering
    \includegraphics[width=0.9\textwidth]{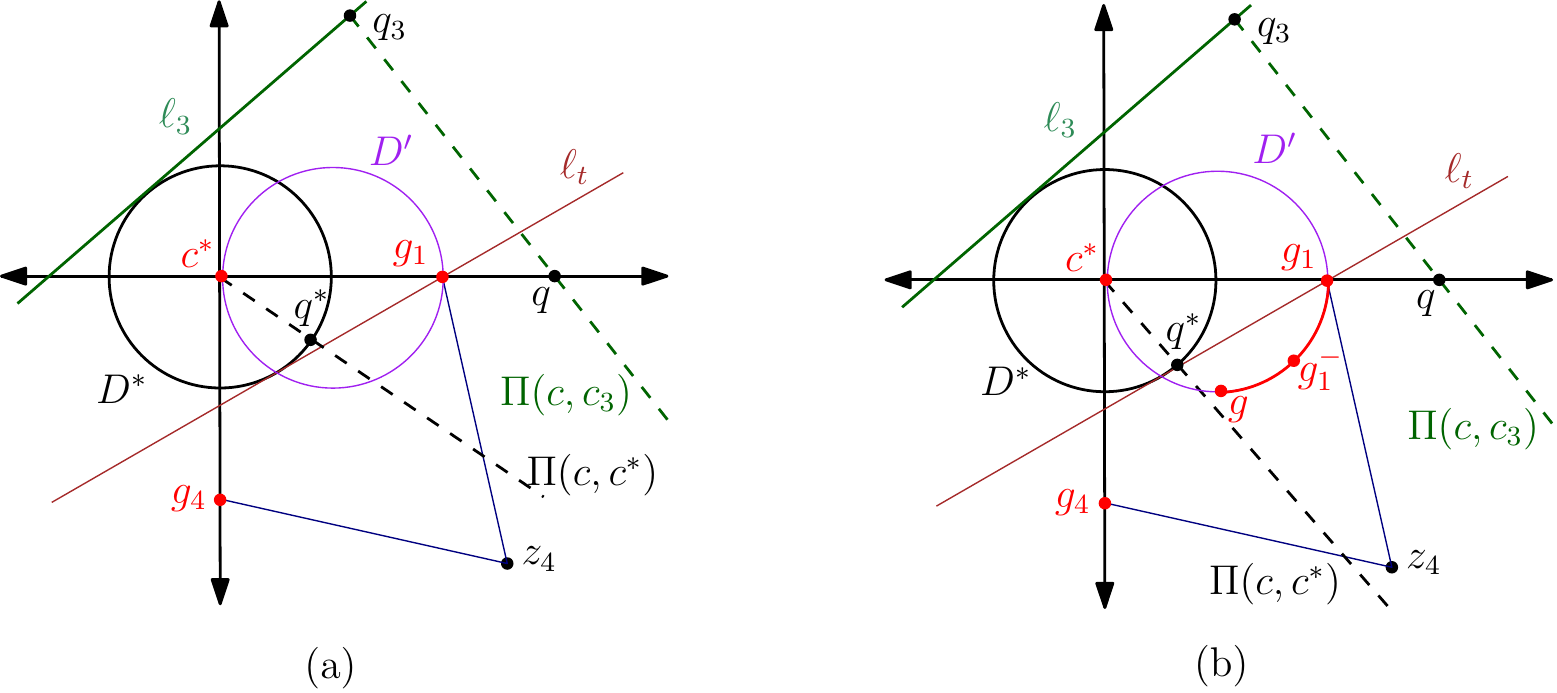}
    \caption{Illustration of the proof of Lemma~\ref{lemma:Q4g1}: (a) $\Pi(c,c^*)$ intersects $\overline{z_4g_1}$, and (b) $\Pi(c,c^*)$ intersects $\overline{z_4g_4}$.}
    \label{fig:Q4g1}
    \end{figure}

The following lemma and its proof is symmetric to Lemma~\ref{lemma:Q4g1}.
\begin{lemma}\label{lemma:Q4g4}
Let $D\in \D$ be a disk centered at $c$ with radius $r$ and $c' \in Q_4$, such that $\Pi(c,c_3)$ intersects the $y$-axis at a point $q$ where $y(q) < y(g_4)$. 
If $\alpha_3 \leq \frac{\pi}{3}$, then
\begin{itemize}
    \item if $\Pi(c,c^*)$ intersects $\overline{z_4g_4}$, then $|\Pi(c,g_4)|\le r$; see Figure~\ref{fig:Q4g4}(a); and
    \item if $\Pi(c,c^*)$ intersects $\overline{z_4g_1}$, then $|\Pi(c,g_4)|\le r$ and $|\Pi(c,g_4^+)|\le r$; see Figure~\ref{fig:Q4g4}(b).
\end{itemize}
\end{lemma}
    \begin{figure}[htb]
    \centering
    \includegraphics[width=0.9\textwidth]{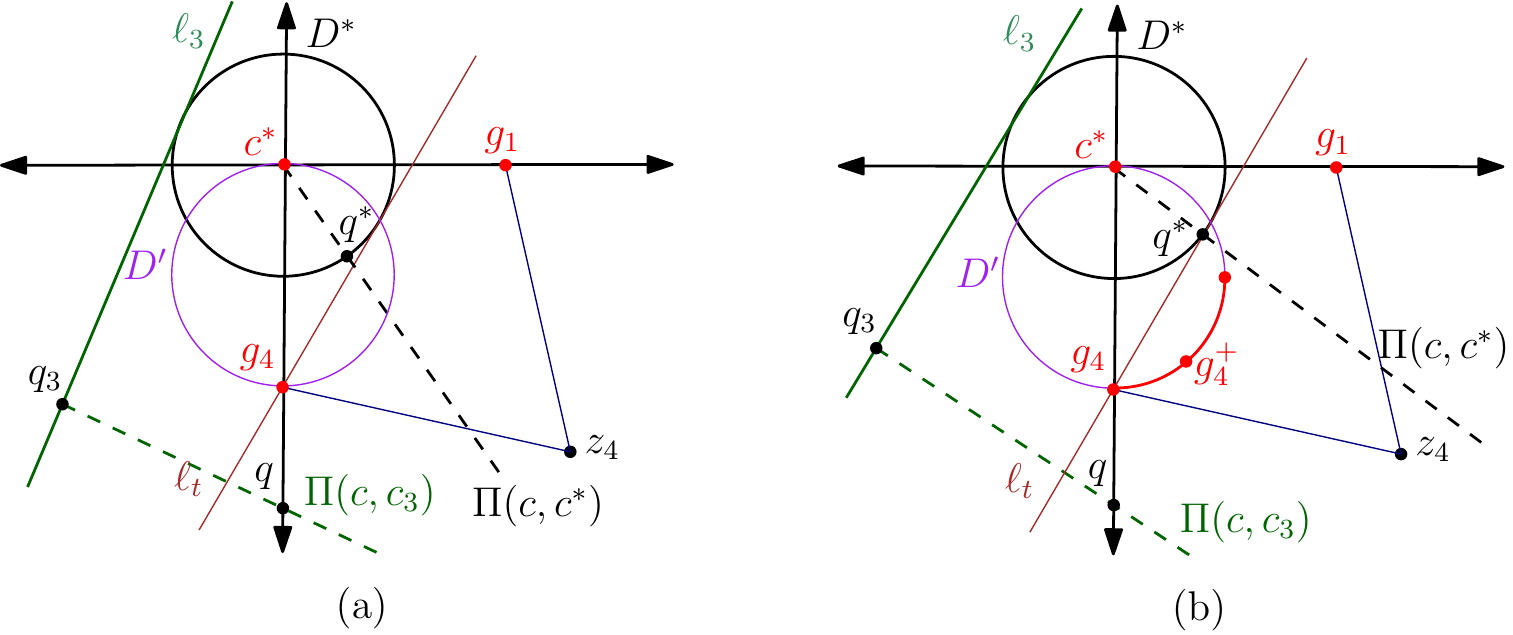}
    \caption{Illustration of Lemma~\ref{lemma:Q4g4}: (a) $\Pi(c,c^*)$ intersects $\overline{z_4g_4}$, and (b) $\Pi(c,c^*)$ intersects $\overline{z_4g_1}$.}
    \label{fig:Q4g4}
    \end{figure}
    %
%

\section{The Algorithm}
In this section, we show how to compute a set $S$ of five points that pierce all the disks of $\D$. 
The algorithm, in a high-level description, works as follows.
It first initializes $S$ by $\{c^*\}$.
Then, it goes over the segments $\overline{z_ig_{i}}$, $\overline{z_ig_{i+1}}$, for each $i=1,2,3,4$ (in a fixed order), and, for each segment, it checks whether the polygon intersects the segment, and adds to $S$ a point $g'_i \in \{g_i,g^+_i,g^-_i\}$.

Recall that $\alpha_2$ (resp., $\alpha_3$) is the acute angle between $\l_2$ (resp., $\l_3$) and the $x$-axis, and notice that at most one of them is greater than $\frac{\pi}{3}$.
We distinguish between three cases:
\begin{itemize}
	\item[(i)] $\alpha_2 > \frac{\pi}{3}$;
	\item[(ii)] $\alpha_3 > \frac{\pi}{3}$;
	\item[(iii)] $\alpha_2 \le \frac{\pi}{3}$ and $\alpha_3 \le \frac{\pi}{3}$.
\end{itemize}
Notice that Case (i) and Case (ii) are symmetric. In Algorithm~\ref{alg:ComputeSii}, we describe how to compute $S$ in Case (i), and, in Algorithm~\ref{alg:ComputeSiii}, we describe how to compute $S$ in Case (iii).  

\floatname{algorithm}{Algorithm}
\begin{algorithm}[htb]
\caption{\emph{Compute S when $\alpha_2 > \frac{\pi}{3}$}} \label{alg:ComputeSii}


\begin{algorithmic}[1]
\STATE $g'_1 \leftarrow g_1$, $g'_2 \leftarrow g_2$, $g'_3 \leftarrow g_3$, $g'_4 \leftarrow g_4$ \\

\STATE \textbf{if} $P$ does not intersect $\overline{z_1g_1}$ \textbf{then} \\ 
    \STATE \quad \ \textbf{if} $P$ intersects $\overline{z_1g_2}$ or $\overline{z_2g_2}$ \textbf{then} \\
        \STATE \quad \ \quad \ $g'_1 \leftarrow $ $g^+_1$ \\
        \STATE \quad \ \quad \ \textbf{if} $P$ intersects $\overline{z_2g_2}$ \textbf{then} \\
            \STATE \quad \ \quad \ \quad \ $g'_2 \leftarrow g^+_2$ \\
    \STATE \quad \ \textbf{else} \\
        \STATE \quad \ \quad \ \textbf{if} $P$ intersects $\overline{z_4g_4}$ \textbf{then} \\ 
            \STATE \quad \ \quad \ \quad \ $g'_1 \leftarrow $ $g^-_1$   \\

\STATE \textbf{if} $P$ does not intersect $\overline{z_2g_3}$ \textbf{then} \\ 
    \STATE \quad \ \textbf{if} $P$ intersects $\overline{z_2g_2}$ \textbf{then} \\
        \STATE \quad \ \quad \ $g'_2 \leftarrow g^+_2$ \\

\STATE \textbf{if} $P$ does not intersect $\overline{z_3g_4}$ \textbf{then} \\ 
    \STATE \quad \ \textbf{if} $P$ intersects $\overline{z_3g_3}$ \textbf{then} \\
        \STATE \quad \ \quad \ $g'_4 \leftarrow $ $g^-_4$ \\
       
\STATE \textbf{return} $S=\{c^*,g'_1,g'_2,g'_3,g'_4\}$ 

\end{algorithmic}
\end{algorithm}

\floatname{algorithm}{Algorithm}
\begin{algorithm}[H]
\caption{\emph{Compute S when $\alpha_2 \le \frac{\pi}{3}$ and $\alpha_3 \le \frac{\pi}{3}$}} \label{alg:ComputeSiii}

\begin{algorithmic}[1]
\STATE $g'_1 \leftarrow g_1$, $g'_2 \leftarrow g_2$, $g'_3 \leftarrow g_3$, $g'_4 \leftarrow g_4$ \\

\STATE \textbf{if} $P$ does not intersect $\overline{z_1g_1}$ \textbf{then} \\ 
    \STATE \quad \ \textbf{if} $P$ intersects $\overline{z_1g_2}$ \textbf{then} \\
        \STATE \quad \ \quad \ $g'_1 \leftarrow g^+_1$ \\
\STATE \textbf{if} $P$ does not intersect $\overline{z_2g_3}$ \textbf{then} \\ 
    \STATE \quad \ \textbf{if} $P$ intersects $\overline{z_2g_2}$ \textbf{then} \\
        \STATE \quad \ \quad \ $g'_3 \leftarrow g^-_3$ \\

\STATE \textbf{if} $P$ does not intersect $\overline{z_1g_1}$, $\overline{z_1g_2}$, $\overline{z_2g_2}$ nor $\overline{z_2g_3}$ \textbf{then} \\ 
    \STATE \quad \ \textbf{if} $P$ intersects $\overline{z_3g_4}$ \textbf{then} \\ 
        \STATE\quad \ \quad \ $g'_3 \leftarrow g^+_3$ \\
    \STATE \quad \ \textbf{if} $P$ intersects $\overline{z_4g_4}$ \textbf{then} \\ 
        \STATE \quad \ \quad \ $g'_1 \leftarrow g^-_1$\\
            
\STATE \textbf{return} $S=\{c^*,g'_1,g'_2,g'_3,g'_4\}$ 

\end{algorithmic}
\end{algorithm}

\section{Correctness}
Let $D\in \D$ be a disk with center $c$ and radius $r$. 
For each $i \in \{1,2,3\}$, let $q_i$ be the intersection point of the path $\Pi(c,c_i)$ with the line $\ell_i$.
Let $q^*$ be the intersection point of the path $\Pi(c,c^*)$ with the boundary of $D^*$, and let $c'$ be the point on $\Pi(c,c^*)$, such that the edge $(c',c^*)$ is the last edge in $\Pi(c,c^*)$. 
That is, $c'$ is the first point on $\Pi(c,c^*)$ that is visible from $c^*$.
We prove that the set $S =\{c^*,g'_1,g'_2,g'_3,g'_4\}$ (that is computed by the algorithm) pierces all the disks of $\D$.
In the proof, we distinguish between three cases: (i) $\alpha_2 > \frac{\pi}{3}$; (ii) $\alpha_3 > \frac{\pi}{3}$; and (iii) $\alpha_2 \le \frac{\pi}{3}$ and $\alpha_3 \le \frac{\pi}{3}$.
Following the algorithm, we show in Section~\ref{casei:>60} the proof for Case (i) and in Section~\ref{caseiii:bothlessThan60} the proof for Case (iii) 
(since Case (i) and Case (ii) are symmetric). 

\subsection{Case (i): $\alpha_2 > \frac{\pi}{3}$}\label{casei:>60}
Let $D\in \D$ be a disk with center $c$ and radius $r$.
We show that $D$ is pierced by at least one of the points of $S$.
We distinguish between four cases according to which quadrant $c'$ belongs to.

\subsubsection{$c' \in Q_1$}~\label{sec:Q1}
We prove that $D$ is pierced by at least one of the points $g'_1$, $g'_2$, or $c^*$.
We distinguish between four cases. \\
\textbf{Case~1:} The polygon does not intersect $\overline{z_1g_1}$, $\overline{z_1g_2}$, $\overline{z_2g_2}$, nor $\overline{z_4g_4}$.
In this case, $g'_1 = g_1$ and $g'_2 = g_2$, and by Lemma~\ref{lemma:NoIntersection}, $D$ is pierced by at least one of the points $g'_1$, $g'_2$, and $c^*$.
\\ \textbf{Case~2:} The polygon intersects $\overline{z_1g_1}$; see Figure~\ref{fig:Q1 Case2}.
In this case, $g'_1 = g_1$. 
Consider the path $\Pi(c,c_1)$ and notice that this path intersects the positive $x$-axis. Let $q$ be this intersection point. Thus, $|\Pi(c,q)| + 1 \le r$. 
\begin{itemize}
    \item[(i)] If $x(c^*) \le x(q) \le x(g_1)$, 
    then by Lemma~\ref{lemma:c*g'1}, Item (i), 
    $D$ contains at least one of the points $c^*$ or $g_1$; see Figure~\ref{fig:Q1 Case2}(a).
    \item[(ii)] If $x(g_1) < x(q)\leq x(z_1)$, then, since $x(g_1)=2$ and $x(z_1) < 3$, we have $|qg_1| < 1$. Since $q$ is the intersection point of $\Pi(c,c_1)$ with the $x$-axis, the polygon does not intersect $\overline{qg_1}$.
    Thus, $|\Pi(c,g_1)| \le |\Pi(c,q)| + |qg_1| \le |\Pi(c,q)| + 1 \le r$. Therefore, $D$ contains $g_1$.
    \item[(iii)] If $x(q) > x(z_1)$, then consider the path $\Pi(c,c^*)$ and let $p$ be the intersection point of $\Pi(c,c^*)$ with $\overline{z_1g_1}$; see Figure~\ref{fig:Q1 Case2}(b). Thus, the polygon does not intersect $\overline{pg_1}$, and, by Observation~\ref{obs:prabola}, we have $|pg_1| \le |pq^*| \le |\Pi(p,q^*)|$. Thus, $|\Pi(c,g_1)| \le |\Pi(c,p)| + |pg_1| \le |\Pi(c,p)| + |pq^*| \le |\Pi(c,p)| + |\Pi(p,q^*)| = |\Pi(c,q^*)| \le r$. Therefore, $D$ contains $g_1$.
\end{itemize}
\begin{figure}[htb]
\centering
\includegraphics[width=0.85\textwidth]{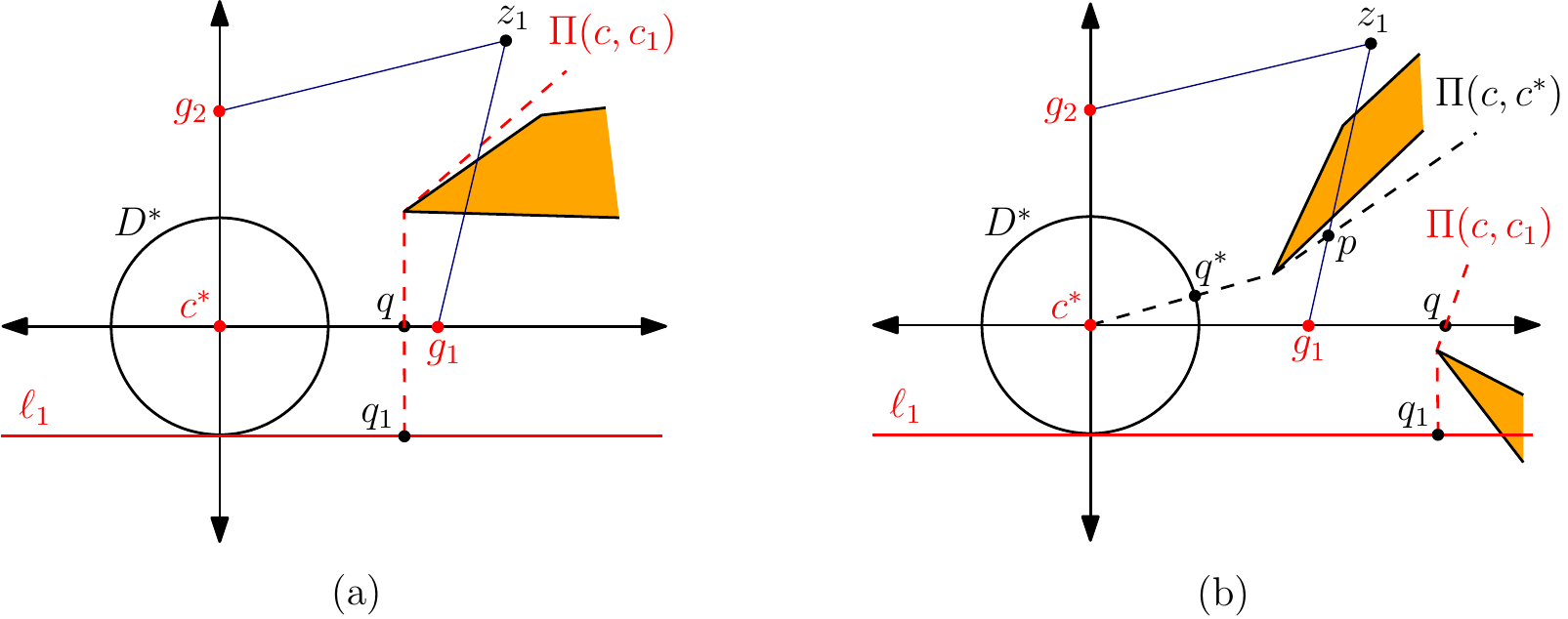}
\caption{Illustration of the proof of Case~2. (a) $x(c^*)\leq x(q)\leq x(z_1)$, and (b) $x(q) > x(z_1)$.}
\label{fig:Q1 Case2}
\end{figure}

\noindent
\textbf{Case~3:} The polygon does not intersect $\overline{z_1g_1}$ but intersects $\overline{z_1g_2}$ 
or $\overline{z_2g_2}$.
In this case, $g_1'=g^+_1$.
Consider the path $\Pi(c,c_1)$ and notice that this path intersects the positive $x$-axis at a point $q$.
\begin{itemize}
    \item[(i)] If $x(c^*)\leq x(q) \leq x(g^+_1)$, then, by Lemma~\ref{lemma:c*g'1}, Item (i), $|\Pi(c,c^*)| \le r$ or  $|\Pi(c,g^+_1)| \le r$, and therefore $D$ contains  $c^*$ or $g^+_1$.
    \item[(ii)] If $x(q) > x(g^+_1)$ and $\Pi(c,c^*)$ intersects the segment $\overline{z_1g_2}$, then, 
    \begin{itemize}
        \item if the polygon intersects the segment $\overline{z_2g_2}$, then, by Lemma~\ref{lemma:g1+ right z1g2}, Item (i), $|\Pi(c,g^+_1)| \le r$, and therefore $D$ contains $g^+_1$; and
        \item if the polygon does not intersect the segment $\overline{z_2g_2}$, then, in this case, $g'_2=g_2$, and, by Lemma~\ref{lemma:g1+ right z1g2}, Item (ii), $|\Pi(c,g^+_1)| \le r$ or $|\Pi(c,g_2)| \le r$, and therefore $D$ contains  $g^+_1$ or $g_2$. 
    \end{itemize}
    \item[(iii)] If $x(q) > x(g^+_1)$ and $\Pi(c,c^*)$ intersects the segment $\overline{z_1g_1}$, then, by Lemma~\ref{lemma:g1+ right z1g1}, Item (i), $|\Pi(c,g^+_1)| \le r$, and therefore $D$ contains $g^+_1$.
    \item[(iv)] If $x(q) > x(g^+_1)$ and $\Pi(c,c^*)$ does not intersect the segments $\overline{z_1g_1}$ nor $\overline{z_1g_2}$, then, by Lemma~\ref{lemma:not intersecting z1g1 z1g2}, Item (i), $|\Pi(c,c^*)| \le r$ or $|\Pi(c,g^+_1)| \le r$, and therefore $D$ contains  $c^*$ or $g^+_1$.
\end{itemize}    
%
\textbf{Case~4:} The polygon does not intersect $\overline{z_1g_1}$, $\overline{z_1g_2}$ nor $\overline{z_2g_2}$ but intersects $\overline{z_4g_4}$. In this case, $g_1'=g^-_1$, and thus, by Lemma~\ref{lemma:g1-}, Item (i), 
$|\Pi(c,g^-_1)| \le r$ or $|\Pi(c,c^*)| \le r$, and therefore $D$ contains $g^-_1$ or $c^*$.


\subsubsection{$c' \in Q_2$}~\label{sec:Q2}
We prove that $D$ is pierced by at least one of the points $g'_2$ , $g'_3$, or $c^*$.
We distinguish between three cases. 
\\ \textbf{Case~1:} The polygon does not intersect $\overline{z_2g_2}$ nor $\overline{z_2g_3}$.
In this case, $g'_2 = g_2$ and $g'_3 = g_3$, and by Lemma~\ref{lemma:NoIntersection}, $D$ is pierced by at least one of the points $g'_2$ , $g'_3$, or $c^*$.
\noindent
\\ \textbf{Case~2:} The polygon intersects $\overline{z_2g_3}$. In this case, $g_3'=g_3$ and $D$ contains at least one of the points $c^*$ or $g_3$ (the proof is symmetric to Case~2 in Section~\ref{sec:Q1}); see Figure~\ref{fig:Q2 Case2}.
\begin{figure}[htb]
\centering
\includegraphics[width=0.9\textwidth]{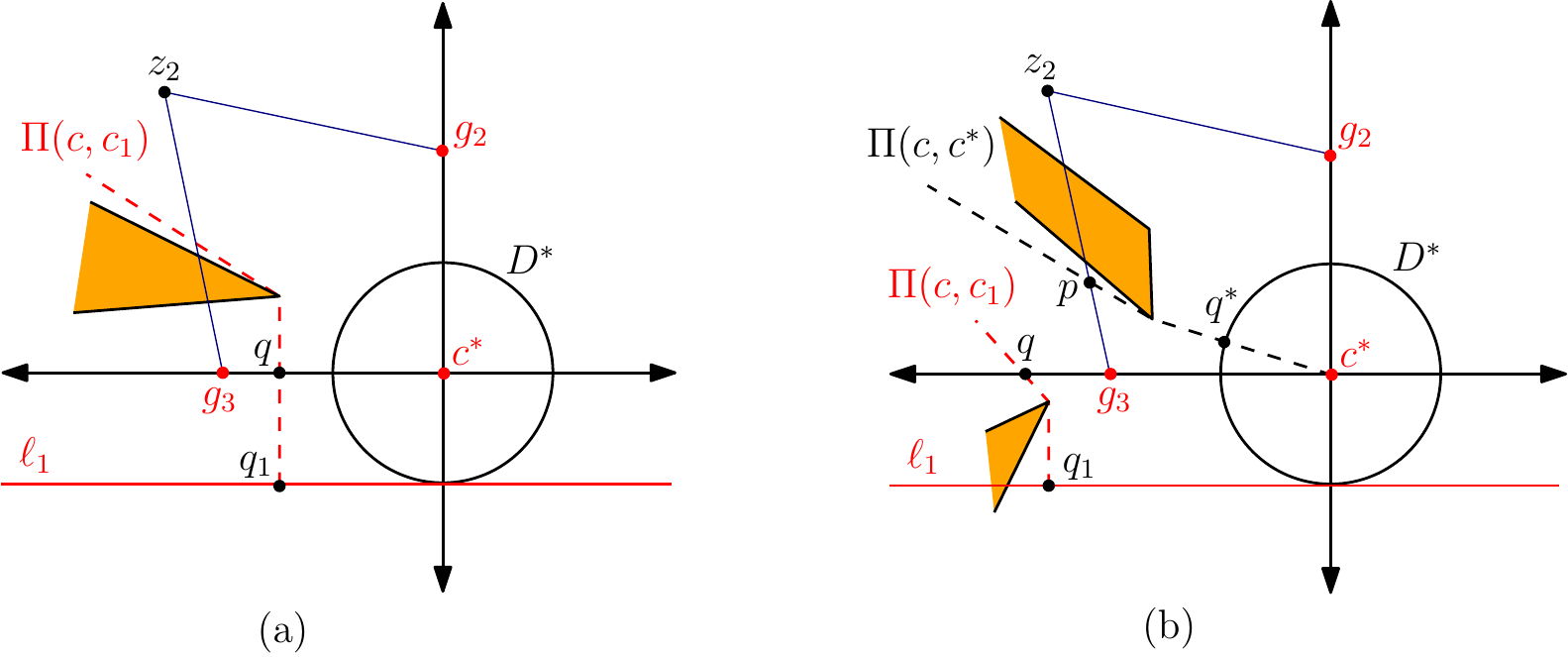}
\caption{Case~2: (a) $x(g_3)\leq x(q)\leq x(c^*)$, and (b) $x(q) < x(z_2)$.}
\label{fig:Q2 Case2}
\end{figure}
\noindent
\\ \textbf{Case~3:} The polygon does not intersect $\overline{z_2g_3}$ but intersects $\overline{z_2g_2}$. In this case, $g_2'=g^+_2=(-1,1)$ and $g_3'=g_3$. Consider the path $\Pi(c,c_1)$ and notice that it intersects the negative $x$-axis at a point $q$. Thus, $|\Pi(c,q)| + 1 \le r$.
\begin{itemize}
    \item[(i)] If $x(g_3) \leq x(q) \leq x(c^*)$, then, by Lemma~\ref{lemma:c*g'1}, Item (ii), we have $|\Pi(c,c^*)| \le r$ or  $|\Pi(c,g_3)| \le r$, and therefore $D$ contains $c^*$ or $g_3$. 
    \item[(ii)] If $x(z_2) \leq x(q) \leq x(g_3)$, then, since $x(z_2) > -3$ and $x(g_3) =-2$, we have $|qg_3| < 1$.
    Since the polygon does intersect $\overline{qg_3}$, we have $|\Pi(c,g_3)| \le |\Pi(c,q)| + |qg_3| < |\Pi(c,q)| + 1 \le r$. Therefore, $D$ contains $g_3$.
    \item[(iii)] If $x(q)< x(z_2)$, then consider the path $\Pi(c,c^*)$, and notice that this path intersects either $\overline{z_2g_2}$ or $\overline{z_2g_3}$. 
    \begin{itemize}
        \item If $\Pi(c,c^*)$ intersects $\overline{z_2g_2}$, then, since $\alpha_2 > \frac{\pi}{3}$, by Lemma~\ref{lemma:g2+}, we have $|\Pi(c,g^+_2)| \le r$, and therefore $D$ contains $g^+_2$.
        \item If $\Pi(c,c^*)$ intersects $\overline{z_2g_3}$, then let $p$ be this intersection point. 
        Since the polygon does not intersect $\overline{pg_3}$, by Observation~\ref{obs:prabola}, we have $|pg_3| \le |pq^*| \le |\Pi(p,q^*)|$. Thus, $|\Pi(c,g_3)| \le |\Pi(c,p)| + |pg_3| \le |\Pi(c,p)| + |\Pi(p,q^*)| = |\Pi(c,q^*)| \le r$. Therefore, $D$ contains $g_3$.
    \end{itemize}
\end{itemize}

\subsubsection{$c'\in Q_3$}~\label{sec:Q3}
We prove that $D$ is pierced by at least one of the points $g'_3$ , $g'_4$, or $c^*$. 
\\ \textbf{Case~1:} The polygon does not intersect $\overline{z_3g_3}$ nor $\overline{z_3g_4}$.
In this case, $g'_3 = g_3$ and $g'_4 = g_4$, and by Lemma~\ref{lemma:NoIntersection}, $D$ is pierced by at least one of the points $g'_3$ , $g'_4$, and $c^*$.
\noindent
\\\textbf{Case~2:} The polygon intersects $\overline{z_3g_4}$. In this case $g_3'=g_3$ and $g_4'=g_4$. Consider the path $\Pi(c,c_2)$, and notice that it intersects either the negative $y$-axis or the negative $x$-axis. 
Let $q$ be this intersection point. Thus, $|\Pi(c,q)| + 1 \le r$.
\begin{itemize}
    \item[(i)] If $\Pi(c,c_2)$ intersects the negative $y$-axis, then $D$ contains at least one of the points $c^*$ or $g_4$ (the proof is symmetric to Case~2 in Section~\ref{sec:Q1}); see Figure~\ref{fig:Q3Case21}.
    \begin{figure}[htb]
    \centering
    \includegraphics[width=0.85\textwidth]{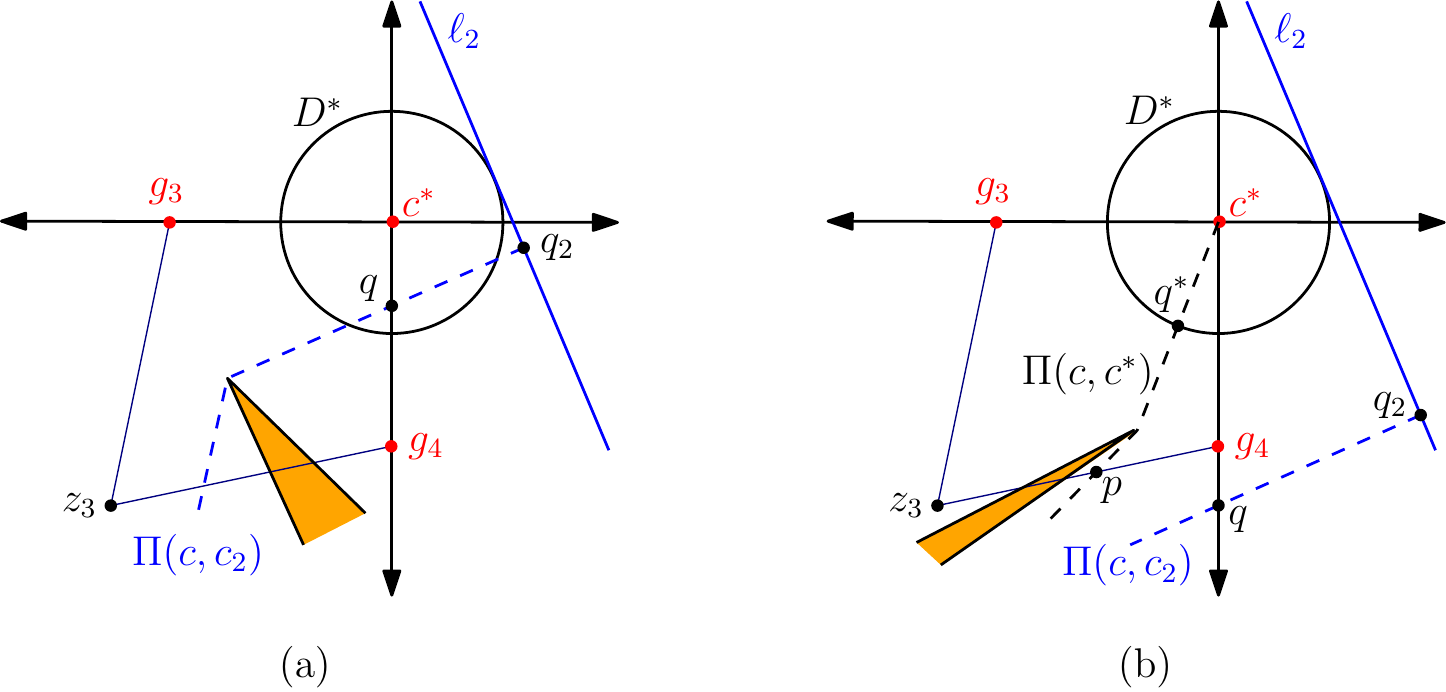}
    \caption{$\Pi(c,c_2)$ intersects the $y$-axis (a) $y(g_4) \leq y(q) \leq y(c^*)$, and (b) $y(q) < y(g_4)$.}
    \label{fig:Q3Case21}
    \end{figure}
    \item[(ii)] If $\Pi(c,c_2)$ intersects the negative $x$-axis and $x(g_3) \leq x(q) \leq x(c^*)$,  
    then, by Lemma~\ref{lemma:c*g'1}, Item (iii),   
    $D$ contains at least one of the points $c^*$ or $g_3$; see Figure~\ref{fig:Q3Case2}(a).
    \item[(iii)] If $\Pi(c,c_2)$ intersects the negative $x$-axis and $x(q) < x(g_3)$, then consider the path  $\Pi(c,c^*)$ and notice that, since $\alpha_2 > \frac{\pi}{3}$ and $x(q) < x(g_3)$, this path intersects $\overline{z_3g_3}$ at a point $p$ and the polygon does not intersect $\overline{pg_3}$; see Figure~\ref{fig:Q3Case2}(b). Hence, by Observation~\ref{obs:prabola}, we have $|pg_3| \le |pq^*| \le |\Pi(p,q^*)|$. Thus, $|\Pi(c,g_3)| \le |\Pi(c,p)| + |pg_3| \le |\Pi(c,p)| + |\Pi(p,q^*)| = |\Pi(c,q^*)| \le r$. Therefore, $D$ contains $g_3$.
\end{itemize}
    \begin{figure}[htb]
    \centering
    \includegraphics[width=0.9\textwidth]{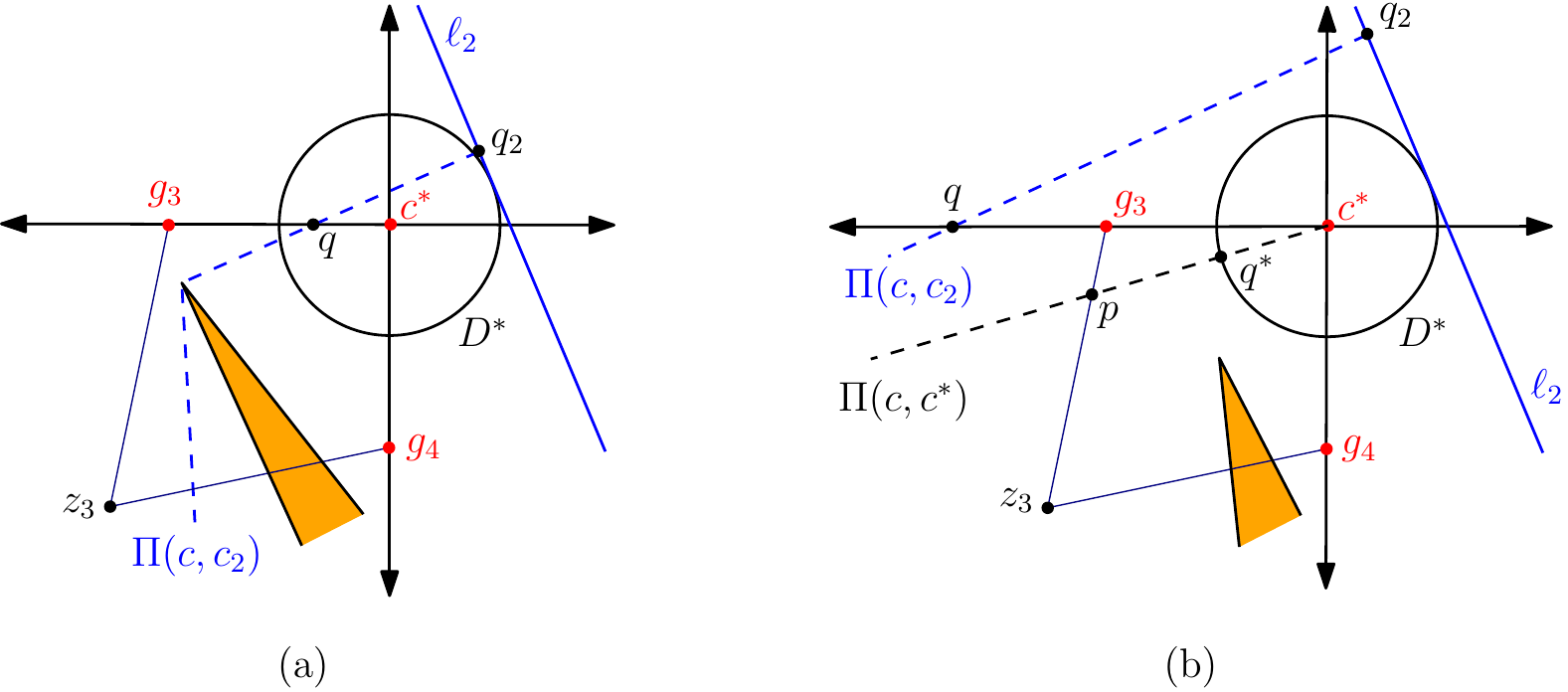}
    \caption{Case~2: (a) $\Pi(c,c_2)$ intersects the $x$-axis and $x(g_3) \leq x(q) \leq x(c^*)$, and (b) $\Pi(c,c_2)$ intersects the $x$-axis and $x(q) < x(g_3)$.}
    \label{fig:Q3Case2}
    \end{figure}
    %
%
\textbf{Case~3:} The polygon does not intersect $\overline{z_3g_4}$ but intersects $\overline{z_3g_3}$. In this case $g_3'=g_3$ and $g_4'=g^-_4$. Consider the path $\Pi(c,c_2)$ and notice that it intersects either the negative $y$-axis or the negative $x$-axis. Let $q$ be this intersection point. Thus, $|\Pi(c,q)| + 1 \le r$.
\begin{itemize}
    \item[(i)] $\Pi(c,c_2)$ intersects the negative $x$-axis, then $D$ contains at least one of the points $c^*$ and $g_3$ (the proof is symmetric to the proof of Items (ii) and (iii) in the previous case); see Figure~\ref{fig:Q3Case31}.
    \begin{figure}[htb]
    \centering
    \includegraphics[width=0.87\textwidth]{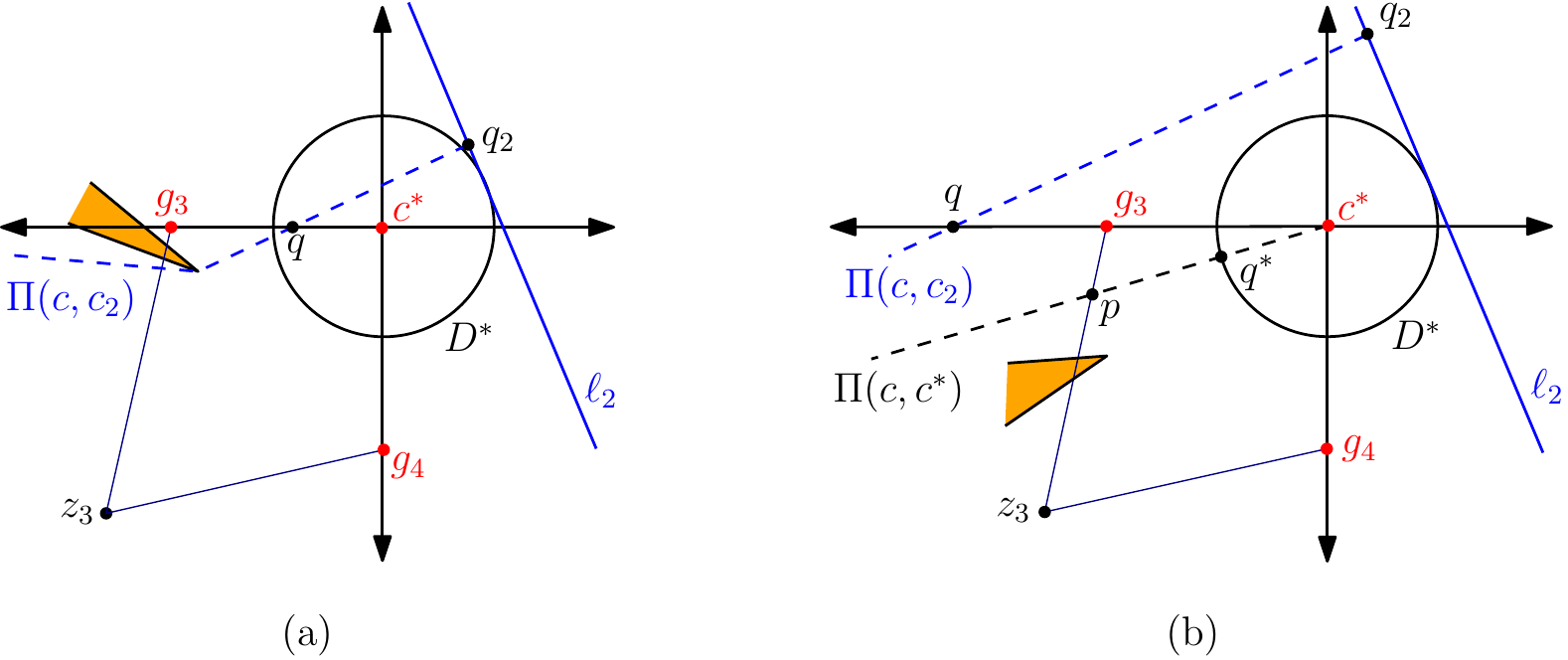}
    \caption{$\Pi(c,c_2)$ intersects the $x$-axis. (a)  $x(g_3) \leq x(q) \leq x(c^*)$, and (b) $x(q) < x(g_3)$.}
    \label{fig:Q3Case31}
    \end{figure}
    \item[(ii)] $\Pi(c,c_2)$ intersects the negative $y$-axis and $y(g^-_4) \leq y(q) \leq y(c^*)$, then, by Lemma~\ref{lemma:c*g'1}, Item (iv), we have $|\Pi(c,c^*)| \le r$ or  $|\Pi(c,g^-_4)| \le r$, and therefore $D$ contains  $c^*$ or $g^-_4$.
    \item[(iii)] If $y(q) < y(g^-_4)$ and $\Pi(c,c^*)$ intersects the segment $\overline{z_3g_3}$, then, by Lemma~\ref{lemma:g1+ right z1g2}, Item (iii), $|\Pi(c,g^-_4)| \le r$ or $|\Pi(c,g_3)| \le r$, and therefore $D$ contains  $g^-_4$ or $g_3$.
    \item[(iv)] If $y(q) < y(g^-_4)$ and $\Pi(c,c^*)$ intersects the segment $\overline{z_3g_4}$, then, by Lemma~\ref{lemma:g1+ right z1g1}, Item (ii), $|\Pi(c,g^-_4)| \le r$, and therefore $D$ contains $g^-_4$.
    \item[(v)] If $y(q) < y(g^-_4)$ and $\Pi(c,c^*)$ does not intersect the segments $\overline{z_3g_4}$ nor $\overline{z_3g_3}$, then, by Lemma~\ref{lemma:not intersecting z1g1 z1g2}, Item (ii), $|\Pi(c,g^-_4)| \le r$, and therefore $D$ contains $g^-_4$.
\end{itemize}

\subsubsection{$c'\in Q_4$}
We prove that $D$ is pierced by at least one of the points $g'_1$, $g'_4$, or $c^*$.
Consider the path $\Pi(c,c_3)$, and notice that it intersects either the positive $x$-axis or the negative $y$-axis at a point $q$. \\
\textbf{The point $q$ is on the positive $x$-axis}.
\\ \textbf{Case~1:} $x(c^*)\leq x(q)\leq x(g_1')$. By Lemma~\ref{lemma:c*g'1}, Item (v),  $|\Pi(c,c^*)|\le r$ or $|\Pi(c,g'_1)|\le r$, and therefore $D$ contains $c^*$ or $g_1'$.
\\ \textbf{Case~2:} $x(q) > x(g_1')$. We distinguish between three cases.
\\ \textbf{Case~2.1:} $g_1'=g_1$. 
    \begin{itemize}
        \item[(i)] If the polygon intersects $\overline{z_1g_1}$, then $\Pi(c,c_3)$ intersects $\overline{z_1g_1}$ at a point $p$. Thus, $g_1$ is inside the pseudo-triangle $\triangle(c,q^*,p)$, and, by Observation~\ref{obs:observationTriangle}, $D$ contains $g_1$.
        \item[(ii)] If the polygon intersects $\overline{z_4g_1}$, then $D$ contains at least one of the points $c^*$ or $g_1$ (the proof is symmetric to the proof of Case~2 in Section~\ref{sec:Q1}). 
        \item[(iii)] If the polygon does not intersect $\overline{z_1g_1}$ nor $\overline{z_4g_1}$, then, since $g_1'=g_1$, the polygon does not intersect $\overline{z_4g_4}$.
        If $g'_4 = g_4$, then, by Lemma~\ref{lemma:NoIntersection}, $D$ is pierced by at least one of the points $g'_1$, $g'_4$, and $c^*$. 
        Otherwise, $g_4'=g_4^-$.
        In this case, the polygon intersects $\overline{z_3g_3}$, and, by Observation~\ref{lemma:ell3Angle}, we have $\alpha_3 > \frac{\pi}{5}$.  Thus, by Lemma~\ref{lemma:Q4g1}, $|\Pi(c,g_1)|\le r$, and therefore $D$ contains $g_1$.
    \end{itemize}
    \textbf{Case~2.2:} $g_1'=g^+_1\neq g_1$. By Lemma~\ref{lemma:g1-}, Item (ii), $|\Pi(c,c^*)|\le r$ or $|\Pi(c,g^+_1)|\le r$, and therefore $D$ contains $c^*$ or $g_1^+$.
    \\ \textbf{Case~2.3:} $g_1'=g_1^-$. 
    \begin{itemize}
        \item[(i)] If $\Pi(c,c^*)$ intersects $\overline{z_4g_4}$, then, if $g'_4 = g_4$, then, by Lemma~\ref{lemma:g1+ right z1g2}, Item (iv), $|\Pi(c,g^-_1)|\leq r$ or $|\Pi(c,g_4)|\leq r$, and therefore $D$ contains  $g_1^-$ or $g_4$.
        Otherwise, $g_4'=g_4^-$. In this case, the polygon intersects $\overline{z_3g_3}$, and, by Observation~\ref{lemma:ell3Angle}, we have $\alpha_3 > \frac{\pi}{5}$.  Thus, by Lemma~\ref{lemma:Q4g1}, $|\Pi(c,g_1)|\le r$, and therefore $D$ contains $g_1$.
        \item[(ii)] If $\Pi(c,c^*)$ intersects $\overline{z_4g_1}$, then by Lemma~\ref{lemma:g1+ right z1g1}, Item (iii), $|\Pi(c,g^-_1)|\leq r$, and therefore $D$ contains $g_1^-$.
        \item[(iii)] If $\Pi(c,c^*)$ does not intersect $\overline{z_4g_1}$ nor $\overline{z_4g_4}$, then, by Lemma~\ref{lemma:not intersecting z1g1 z1g2}, Item (iii), $|\Pi(c,c^*)| \le r$ or $|\Pi(c,g^-_1)| \le r$, and therefore $D$ contains $c^*$ or $g^-_1$.
    \end{itemize}
\textbf{The point $q$ is on the negative $y$-axis}. 
\\ \textbf{Case~1:} $y(g_4')\leq y(q)\leq y(c^*)$. By Lemma~\ref{lemma:c*g'1}, Item (vi),  $|\Pi(c,c^*)|\le r$ or $|\Pi(c,g'_4)|\le r$, and therefore $D$ contains $c^*$ or $g_4'$.
\\ \textbf{Case~2:} $y(q) < y(g_4')$. We distinguish between two cases. \\
 \textbf{Case~2.1:} $g_4'=g^-_4\neq g_4$. By Lemma~\ref{lemma:g1-}, Item (iii),  $|\Pi(c,c^*)|\le r$ or $|\Pi(c,g^-_4)|\le r$, and therefore $D$ contains $c^*$ or $g^-_4$. \\
\textbf{Case~2.2:} $g_4'= g_4$. Then, since $\alpha_3 \le \frac{\pi}{3}$, by Lemma~\ref{lemma:Q4g4}, $|\Pi(c,g_4)|\le r$, and therefore $D$ contains $g_4$.

\subsection{Case (iii): $\alpha_2 \leq \frac{\pi}{3}$ and $\alpha_3 \leq \frac{\pi}{3}$}\label{caseiii:bothlessThan60}
Let $D\in\D$ be a disk with center $c$ and radius $r$. We show that $D$ is pierced by at least one of the points of $S$. 
Notice that in Algorithm~\ref{alg:ComputeSiii}, $Q_1$ is symmetric to $Q_2$ and $Q_3$ is symmetric to $Q_4$. Therefore, we show the correctness for the cases where $c' \in Q_1$ and $c' \in Q_4$. 

\subsubsection{$c'\in Q_1$}

We prove that $D$ is pierced by at least one of the points $g'_1$, $g'_2$, or $c^*$.
We distinguish between four cases. \\
\\ \textbf{Case~1:} The polygon does not intersect $\overline{z_1g_1}$, $\overline{z_1g_2}$, nor $\overline{z_4g_4}$.
In this case, $g'_1 = g_1$ and $g'_2 = g_2$, and by Lemma~\ref{lemma:NoIntersection}, $D$ contains at least one of the points $g'_1$, $g'_2$, and $c^*$.
\\ \textbf{Case~2:} The polygon intersects $\overline{z_1g_1}$; see Figure~\ref{fig:Q1 Case2}.
In this case, $g'_1 = g_1$, and $D$ contains at least one of the points $c^*$ or $g_1$. The proof is the same as in Case~2 of Section~\ref{sec:Q1}. 
\noindent
\\ \textbf{Case~3:} The polygon does not intersect $\overline{z_1g_1}$, but intersects $\overline{z_1g_2}$. In this case, $g_1'=g^+_1$ and $g_2'=g_2$, and $D$ contains at least one of the points $g_2$, $g^+_1$ and $c^*$. The proof is the same as in Case~3 of Section~\ref{sec:Q1}. 
\noindent
\\ \textbf{Case~4:} The polygon does not intersect $\overline{z_1g_1}$ nor $\overline{z_1g_2}$ but intersects $\overline{z_4g_4}$. In this case $g_1'=g^-_1$ and $g_2'=g_2$, and $D$ contains at least one of the points $c^*$ or $g^-_1$. The proof is the same as in Case~4 of Section~\ref{sec:Q1}.

\subsubsection{$c'\in Q_4$}
We prove that $D$ is pierced by at least one of the points $g'_1$, $g'_4$, and $c^*$. Notice that in this case where both $\alpha_2$ and $\alpha_3$ are less or equal to $\frac{\pi}{3}$, Algorithm~\ref{alg:ComputeSiii} does not change $g_4$, thus $g_4'=g_4$.
Consider the path $\Pi(c,c_3)$, and notice that it intersects either the positive $x$-axis or the negative $y$-axis at a point $q$. \\
\textbf{The point $q$ is on the positive $x$-axis}.
\\ \textbf{Case~1:} $x(c^*)\leq x(q)\leq x(g_1')$. By Lemma~\ref{lemma:c*g'1}, Item (v),  $|\Pi(c,c^*)|\le r$ or $|\Pi(c,g'_1)|\le r$, and therefore $D$ contains $c^*$ or $g_1'$.
\\ \textbf{Case~2:} $x(q) > x(g_1')$. We distinguish between three cases.
\\ \textbf{Case~2.1:} $g_1'=g_1$.
    \begin{itemize}
        \item[(i)] If the polygon intersects $\overline{z_1g_1}$, then $\Pi(c,c_3)$ intersects $\overline{z_1g_1}$ at a point $p$. Thus, $g_1$ is inside the pseudo-triangle $\triangle(c,q^*,p)$, and, by Observation~\ref{obs:observationTriangle}, $D$ contains $g_1$.
        \item[(ii)] If the polygon intersects $\overline{z_4g_1}$, then $D$ contains at least one of the points $c^*$ or $g_1$ (the proof is symmetric to the proof of Case~2 in Section~\ref{sec:Q1}). 
        \item[(iii)] If the polygon does not intersect $\overline{z_1g_1}$ nor $\overline{z_4g_1}$, then, since $g_1'=g_1$, the polygon does not intersect $\overline{z_4g_4}$.
        Since $g'_4 = g_4$, by Lemma~\ref{lemma:NoIntersection}, $D$ is pierced by at least one of the points $g'_1$, $g'_4$, and $c^*$. 
    \end{itemize}
    \textbf{Case~2.2:} $g_1'=g^+_1\neq g_1$. By Lemma~\ref{lemma:g1-}, Item (ii), $|\Pi(c,c^*)|\le r$ or $|\Pi(c,g^+_1)|\le r$, and therefore $D$ contains $c^*$ or $g_1^+$.
    \\ \textbf{Case~2.3:} $g_1'=g_1^-$. 
    \begin{itemize}
        \item[(i)] If $\Pi(c,c^*)$ intersects $\overline{z_4g_4}$, then, since $g'_4 = g_4$, by Lemma~\ref{lemma:g1+ right z1g2}, Item (iv), we have $|\Pi(c,g^-_1)|\leq r$ or $|\Pi(c,g_4)|\leq r$, and therefore $D$ contains  $g_1^-$ or $g_4$.
        \item[(ii)] If $\Pi(c,c^*)$ intersects $\overline{z_4g_1}$, then by Lemma~\ref{lemma:g1+ right z1g1}, Item (iii), $|\Pi(c,g^-_1)|\leq r$, and therefore $D$ contains $g_1^-$.
        \item[(iii)] If $\Pi(c,c^*)$ does not intersect $\overline{z_4g_1}$ nor $\overline{z_4g_4}$, then, by Lemma~\ref{lemma:not intersecting z1g1 z1g2}, Item (iii), $|\Pi(c,c^*)| \le r$ or $|\Pi(c,g^-_1)| \le r$, and therefore $D$ contains $c^*$ or $g^-_1$.
    \end{itemize}
\textbf{The point $q$ is on the negative $y$-axis}. \\
\textbf{Case~1:} $y(g_4')\leq y(q)\leq y(c^*)$. By Lemma~\ref{lemma:c*g'1}, Item (vi),  $|\Pi(c,c^*)|\le r$ or $|\Pi(c,g'_4)|\le r$, and therefore $D$ contains $c^*$ or $g_4'$. \\ 
\textbf{Case~2:} $y(q) < y(g_4')$. Since, $g_4'= g_4$ and $\alpha_3 \le \frac{\pi}{3}$, by Lemma~\ref{lemma:Q4g4},
we have $|\Pi(c,g_4)|\le r$, and therefore $D$ contains $g_4$.

\section{Conclusion}
We have shown that five points are sufficient to pierce a set of pairwise intersecting geodesic disks inside a polygon $P$. This improves the upper bound of 14, which was provided by Bose et al.~\cite{Bose21}. 
This upper bound is very close to the lower bound for stabbing pairwise intersecting disks in the plane, which was proven to be four.

\bibliographystyle{plain}
\bibliography{ref}

\begin{thebibliography}{1}

\bibitem{Bose21}
P.~Bose, P.~Carmi, and T.~C. Shermer.
\newblock Piercing pairwise intersecting geodesic disks.
\newblock {\em Computat. Geom.}, 98:101774, 2021.

\bibitem{Carmi18}
P.~Carmi, M.~J. Katz, and P.~Morin.
\newblock Stabbing pairwise intersecting disks by four points.
\newblock {\em CoRR}, abs/1812.06907, 2018.

\bibitem{Danzer86}
L.~Danzer.
\newblock Zur l\"{o}sung des {G}allaischen problems \"{u}ber kreisscheiben in
  der {E}uklidischen ebene.
\newblock {\em Studia Sci. Math. Hungar}, 21(1-2):111--134, 1986.

\bibitem{HarPeled21}
S.~Har-Peled, H.~Kaplan, W.~Mulzer, L.~Roditty, P.~Seiferth, M.~Sharir, and
  M.~Willert.
\newblock Stabbing pairwise intersecting disks by five points.
\newblock {\em Discrete Math.}, 344(7):112403, 2021.

\bibitem{Helly23}
E.~Helly.
\newblock \"{U}ber mengen konvexer k\"{o}rper mit gemeinschaftlichen punkten.
\newblock {\em Jahresber. Dtsch. Math.-Ver.}, 32:175--176, 1923.

\bibitem{Helly30}
E.~Helly.
\newblock \"{U}ber systeme von abgeschlossenen mengen mit gemeinschaftlichen
  punkten.
\newblock {\em Monatshefte Math.}, 37(1):281--302, 1930.

\bibitem{Pollack89}
R.~Pollack, M.~Sharir, and G.~Rote.
\newblock Computing the geodesic center of a simple polygon.
\newblock {\em Discrete Comput. Geom.}, 4:611--626, 1989.

\bibitem{Stacho65}
L.~Stacho.
\newblock \"{U}ber ein problem f\"{u}r kreisscheiben familien.
\newblock {\em Acta Sci. Math. (Szeged)}, 26:273--282, 1965.

\bibitem{Stacho814}
L.~Stacho.
\newblock A solution of {G}allai’s problem on pinning down circles.
\newblock {\em Mat. Lapok}, 32(1-3):19--47, 1981/84.

\end{thebibliography}

\end{document}